\documentclass[11pt ]{article}
\usepackage{amsfonts}
\usepackage{amsmath}
\usepackage{amssymb}
\usepackage{enumerate}
\usepackage{natbib}
\usepackage{color}
\usepackage{graphicx,subfigure}
\usepackage{epstopdf}
\usepackage{setspace}
\usepackage{hyperref}
\usepackage[normalem]{ulem} 
\usepackage{soul} 
\usepackage{multirow}
\newcommand\E{\ensuremath{\mathbb{E}}}
\newcommand\R{\ensuremath{\mathbb{R}}}

\newcommand\F{\mathcal{F}}

\newcommand\half{\frac{1}{2}}
\newcommand{\dx}[1]{ \,d#1 }

\newcommand{\indic}[1]{1\hspace{-2.1mm}{1}_{\{#1\}}} 

\DeclareMathOperator*{\argmax}{arg\,max}

\def\R{{\mathbb R}}

\def\P{{\mathbb P}}

\def\1{{\mathbf 1}}
\def\F{{\mathcal F}}

\def\L{{\mathcal L}}

\def\setT{{\mathcal T}}
\newtheorem{theorem}{Theorem}

\newtheorem{corollary}[theorem]{Corollary}

\newtheorem{lemma}[theorem]{Lemma}

\newtheorem{proposition}[theorem]{Proposition}
\newtheorem{remark}[theorem]{Remark}

\newenvironment{proof}[1][Proof]{\noindent\textbf{#1.} }{\ \rule{0.5em}{0.5em}}
\numberwithin{equation}{section}
\numberwithin{theorem}{section}
\begin{document}
\title{Optimal Mean Reversion Trading\\ with Transaction Costs and Stop-Loss Exit\thanks{Work partially supported by NSF grant DMS-0908295.}}
\author{Tim Leung\thanks{Corresponding author. IEOR Department, Columbia University, New York, NY 10027; email:\,\mbox{leung@ieor.columbia.edu}.} \and Xin Li\thanks{IEOR Department, Columbia University, New York, NY 10027; email:\,\mbox{xl2206@columbia.edu}.} }
\date{\today}
\maketitle
\begin{abstract}
Motivated by the industry practice of pairs trading, we study the optimal timing strategies for trading a mean-reverting price spread. An optimal double stopping problem is formulated to analyze the timing to start and subsequently liquidate the position subject to transaction costs. Modeling the price spread by an Ornstein-Uhlenbeck process, we apply a probabilistic methodology and  rigorously derive the optimal price intervals  for  market entry and exit. As an extension, we incorporate a stop-loss constraint to limit  the maximum  loss.  We show that the entry region is characterized by a bounded price interval that lies strictly above the stop-loss level. As for the exit timing, a higher stop-loss level always implies a  lower optimal take-profit level.  Both analytical and numerical results are provided to illustrate the dependence of timing strategies on model parameters such as transaction costs and stop-loss level.
\end{abstract}\vspace{10pt}

\begin{small}
\end{small}


\newpage
\linespread{1.1}

\section{Introduction}

It has been widely observed that many asset prices exhibit mean reversion, including  commodities (see \cite{schwartz1997stochastic}),   foreign exchange rates (see \cite{engel1989long,anthony1998mean,larsen2007diffusion}),  as well as US and global  equities (see  \cite{poterba1988mean,malliaropulos1999mean,balvers2000mean,
gropp2004mean}).  Mean-reverting processes are also used to model  the dynamics of interest rate, volatility, and default risk.  In industry,  hedge fund managers and investors often attempt to construct mean-reverting prices  by simultaneously taking positions in two highly correlated or co-moving assets. The advent of exchange-traded funds (ETFs) has further facilitated this  \emph{pairs trading} approach since some ETFs  are designed to track identical or similar  indexes and assets.  For instance, \cite{Trian_Montana2011}  investigate  the  mean-reverting spreads between  commodity ETFs and design model for statistical arbitrage.  \cite{GoldMinerSpreads2013} also examine the mean-reverting spread  between physical gold and  gold equity ETFs.

Given  the  price dynamics of some risky asset(s), one important problem commonly faced by  individual and institutional investors  is to determine when to open and close a position.  While observing the prevailing market prices, a speculative  investor  can choose to enter the market immediately or wait for a future opportunity. After completing the first trade, the investor will need  to decide when is best  to close  the position. This motivates the investigation of the  optimal sequential timing of trades.
 

In this paper, we study  the optimal timing of trades subject to transaction costs under  the   Ornstein-Uhlenbeck (OU) model. Specifically, our  formulation leads to     an   \emph{optimal double stopping} problem  that gives the optimal entry and exit decision rules.   We obtain analytic solutions for both the entry and exit problems.  In addition, we  incorporate a stop-loss constraint to our trading problem.
 We find that  a higher stop-loss level induces the investor to  voluntarily liquidate earlier at a lower take-profit level.  Moreover,  the entry region is  characterized by a bounded price interval that lies strictly  above 
 stop-loss level. In other words, it is optimal to  wait if  the current price is too high or too close to the lower stop-loss level. This is intuitive since entering the market close to stop-loss implies a high chance of exiting at a loss afterwards. As a result, the delay region (complement of the entry region)  is \emph{disconnected}.  Furthermore, we show that  optimal liquidation level
  decreases with the stop-loss level until they coincide, in which case immediate liquidation is optimal at all price levels.


A typical solution  approach for optimal stopping problems driven by diffusion involves the analytical and numerical  studies of  the associated free boundary problems or variational inequalities (VIs); see e.g.  \cite{Bensoussan}, \cite{Oksendal2003}, and \cite{sun1992nested}.  For our   double optimal stopping problem, this method would determine the value functions from a  pair of VIs and require  regularity  conditions to guarantee that the solutions to the VIs indeed correspond to the optimal stopping problems. As noted by  \cite{dayanik2008optimal}, ``the  variational methods become challenging when the form of the reward function and/or the
dynamics of the diffusion obscure the shape of the optimal continuation region."   In our optimal entry timing problem, the reward function  involves the value function from the exit  timing problem, which is not monotone and can be positive and negative.

In contrast to the variational inequality approach, our proposed  methodology starts with a  characterization  of the   value functions    as the smallest concave majorant of any given reward function.   A key feature of this approach is that  it  allows us to directly construct the value function, without \emph{a priori} finding a candidate value function or imposing conditions on the stopping and delay (continuation) regions, such as whether they are connected or not.  In other words, our method will derive  the structure of the stopping and delay regions as an output.



 Our main results provide   the analytic expressions for  the value functions of the double stopping problems; see Theorems \ref{thm:optLiquOU} and \ref{thm:optEntryOU} (without stop-loss), and Theorems \ref{thm:optLiquOUSL} and \ref{thm:optEntryOUSL} (with stop-loss). 
 In earlier   studies, \cite{dynkin1969theorems} analyze the   concave characterization of excessive functions for  a standard Brownian motion, and \cite{dayanik2003optimal} and \cite{dayanik2008optimal} apply this idea to study   the optimal single  stopping of a  one-dimensional diffusion.   In this regard, we contribute to  this line of work by solving  a number of  optimal double stopping problems with and without a stop-loss exit under the OU model.

Among other related studies, \cite{Ekstrom2011} analyze the optimal single
liquidation timing  under the OU model  with zero
long-run mean and no transaction cost.  The current paper extends their model in a number of ways.  First, we analyze the optimal entry timing as well as the optimal   liquidation timing.  Our model  allows for  a non-zero long-run mean and transaction costs, along with a  stop-loss level.  \cite{song2009stochastic} propose  a numerical  stochastic approximation scheme to  solve for the optimal    buy-low-sell-high strategies over a finite horizon. Under a similar setting, \cite{zhang2008trading} and   \cite{zhang2010switching}   also
investigate the infinite sequential buying and selling/shorting problem under
exponential OU price dynamics with slippage cost.

In the context of  pairs trading, a number of studies have also considered market timing strategy with two  price levels. For example, \cite{gatev2006pairs} study the historical returns from the buy-low-sell-high strategy where the entry/exit levels are set as  $\pm$1 standard deviation from the long-run mean. Similarly,  \cite{Avellaneda2010}  consider starting and ending a pairs trade based on the spread's distance from its mean. In \cite{elliott2005pairs}, the market entry timing is modeled  by the first passage time of an OU process, followed by an exit at  a fixed finite horizon.   In comparison, rather than assigning \textit{ad hoc}  price levels or fixed trading times, our approach will generate the entry and exit thresholds as  solutions  of an optimal double stopping problem. Considering an exponential OU asset price with zero mean, \cite{bertram2010analytic} numerically computes the optimal enter and exit levels that maximize the expected return per unit time.  \cite{gregory2010} also apply this approach to log-spread following the CIR and GARCH diffusion models. Other timing strategies adopted by practitioners have been discussed in  \cite{Vidyamurthy2004}.

On the other hand,  the  related problem of constructing  portfolios  and hedging with mean reverting asset prices has been    studied. For example,  \cite{benth2005note} study the utility maximization problem that involves dynamically trading an exponential OU underlying asset.   \cite{jurek2007dynamic} analyze a finite-horizon portfolio optimization problem with an OU asset subject to the power utility and  Epstein-Zin recursive utility. 
	\cite{chiu2012dynamic} consider the dynamic trading of co-integrated assets with a mean-variance criterion. \cite{yan2012dynamic} derive  the dynamic  trading strategy for two co-integrated stocks in order to     maximize the expected terminal utility of wealth over a fixed horizon. They simplify  the associated Hamilton-Jacobi-Bellman equation and obtain a closed-form solution. In the stochastic control approach, incorporating transaction costs and  stop-loss exit can potentially limit  model tractability and is not implemented in these studies.

The rest of the paper is structured as follows. We  formulate the optimal trading problem  in Section \ref{sect-overview}, followed by a discussion on   our method of solution in Section \ref{sect-meth}. In Section \ref{sect-OU}, we analytically solve the optimal double stopping problem and examine  the optimal entry and exit strategies.  In Section \ref{sect-stoploss}, we study the trading problem with a stop-loss constraint.    The proofs of all lemmas are provided in the Appendix.

\section{Problem Overview}\label{sect-overview}

In the background, we fix the probability space $(\Omega, \F, \P)$ with the historical probability measure $\P$.   We consider an Ornstein-Uhlenbeck (OU) process driven by the SDE:
\begin{align}\label{XOU} dX_{t}=  \mu( \theta - X_t)\,dt+\sigma \,dB_{t},\end{align}
with constants $\mu, \sigma>0$, $\theta \in \R$, and state space $\R$. Here, $B$ is a standard Brownian motion under $\P$.   Denote by   $\mathbb{F}\equiv(\F_t)_{t\geq 0}$   the filtration  generated by $X$.

\subsection{A Pairs Trading Example}\label{sect-pairs}
Let us discuss a pairs trading example where we model the value of  the resulting position  by an OU process.  The primary  objective is to motivate our trading problem, rather  than proposing new estimation methodologies or empirical studies on pairs trading. For related studies and more details, we  refer to the seminal paper by  \cite{engle1987co}, the books \cite{Hamilton1994time, Tsay2005analysis}, and references therein.

We construct a portfolio by holding   $\alpha$ shares of  a risky asset  $S^{(1)}$ and shorting  $\beta$ shares of another risky asset  $S^{(2)}$, yielding a portfolio value  $X^{\alpha, \beta}_t =\alpha S^{(1)}_t - \beta S^{(2)}_t$ at  time $t \ge 0$.   The pair of assets are selected to form a  mean-reverting portfolio value. In addition, one can adjust the strategy $(\alpha, \beta)$ to enhance the level of mean reversion. For the purpose of testing mean reversion, only the ratio between $\alpha$ and $\beta$  matters, so we can keep  $\alpha$ constant while varying $\beta$ without loss of generality.  For every strategy  $(\alpha, \beta)$, we observe the resulting  portfolio values $(x_i^{\alpha, \beta})_{i = 0, 1, \ldots, n}$ realized over an $n$-day period. We  then apply the method of  maximum likelihood estimation (MLE)  to fit the  observed portfolio values to an OU process and determine  the  model parameters.  Under the OU model,   the conditional probability density of $X_{t_i}$ at time $t_i$ given $x_{i-1}$ at $t_{i-1}$ with time increment  $\Delta t = t_{i}-t_{i-1}$  is given by
\begin{align*}
f^{OU}(x_i|x_{i-1};\theta,\mu,\sigma) = \frac{1}{\sqrt{2\pi\tilde{\sigma}^2}}\exp\left(-\frac{(x_i-x_{i-1}e^{-\mu \Delta t} - \theta(1-e^{-\mu \Delta t}))^2 }{2\tilde{\sigma}^2} \right),
\end{align*}
with the constant
\begin{align*}
\tilde{\sigma}^2 = \sigma^2\frac{1-e^{-2\mu \Delta t}}{2\mu}.
\end{align*}
Using  the observed  values $(x_i^{\alpha, \beta})_{i = 0, 1, \ldots, n}$,   we  maximize the average log-likelihood  defined by
\begin{align}\label{maxlhd}
\ell(\theta,\mu,\sigma|x^{\alpha, \beta}_0,x^{\alpha, \beta}_1,\dots, x^{\alpha, \beta}_n) &:= \frac{1}{n}\sum_{i=1}^{n}\ln f^{OU}\left(x^{\alpha, \beta}_i|x^{\alpha, \beta}_{i-1};\theta,\mu,\sigma \right)\notag\\
&=-\frac{1}{2}\ln(2\pi)-\ln(\tilde{\sigma}) -\frac{1}{2n\tilde{\sigma}^2}\sum_{i=1}^{n}[x^{\alpha, \beta}_i-x^{\alpha, \beta}_{i-1}e^{-\mu \Delta t} - \theta(1-e^{-\mu \Delta t} )]^2,
\end{align}
and denote  by $\hat{\ell}(\theta^*,\mu^*,\sigma^*)$ the maximized  average log-likelihood  over $\theta$, $\mu$, and $\sigma$ for a given strategy $(\alpha, \beta)$.
For any $\alpha$, we choose the strategy $(\alpha, \beta^*)$, where  $\beta^* = \argmax_{\beta} \hat{\ell}(\theta^*,\mu^*,\sigma^*|x^{\alpha, \beta}_0,x^{\alpha, \beta}_1,\dots, x^{\alpha, \beta}_n)$. For example,  suppose we invest $A$ dollar(s) in asset $S^{(1)}$, so   $\alpha = A/S^{(1)}_0$ shares is held. At the same time, we short  $\beta = B/S^{(2)}_0$ shares in $S^{(2)}$, for  $B/A = 0.001, 0.002,\dots, 1$. This way, the sign of   the initial   portfolio value depends on the sign of the difference  $A-B$, which is non-negative.  Without loss of generality, we set $A=1$.

In Figure \ref{fig:MLE}, we illustrate an example based on two  pairs of  exchange-traded funds (ETFs), namely,  the Market Vectors Gold Miners (GDX) and  iShares Silver Trust (SLV)  against 
  the  SPDR Gold Trust (GLD) respectively.  These liquidly traded funds aim to track the price movements of  the NYSE Arca Gold Miners Index (GDX),  silver (SLV), and  gold bullion (GLD) respectively. These   ETF pairs are also used in \cite{Trian_Montana2011} and \cite{GoldMinerSpreads2013} for their statistical and empirical studies on  ETF pairs trading.

 Using price data from August 2011 to May 2012 ($n=200$, $\Delta t = 1/252$),  we compute and plot  in Figure \ref{fig:loglikeli} the average log-likelihood against the cash amount $B$, and find  that $\hat{\ell}$ is maximized at $B^* = 0.454$ (resp. 0.493) for the GLD-GDX pair (resp. GLD-SLV pair). From this MLE-optimal $B^*$, we obtain the   strategy $(\alpha, \beta^*)$, where $\alpha = 1/S^{(1)}_0$ and $\beta^* = B^* / S^{(2)}_0$.  In this example, the  average log-likelihood for the GLD-SLV pair happens to dominate that for GLD-GDX, suggesting   a higher degree of fit to the OU model. Figure \ref{fig:path} depicts the historical price paths with the strategy  $(\alpha, \beta^*)$.


We summarize the estimation  results in Table \ref{tab:calib}. For each pair, we first estimate the parameters for the OU model from empirical price data. Then, we use the estimated parameters to simulate  price paths according the corresponding OU process. Based on these   simulated OU paths, we perform another MLE  and obtain another set of OU  parameters as well as the maximum average log-likelihood $\hat{\ell}$. As we can see,  the two sets of estimation outputs (the rows names ``empirical" and ``simulated") are very close, suggesting the  empirical  price process fits well to the OU model.

%
%


%
%

\begin{figure}[ht]
 \centering
 \subfigure[]{
  \includegraphics[scale=0.58]{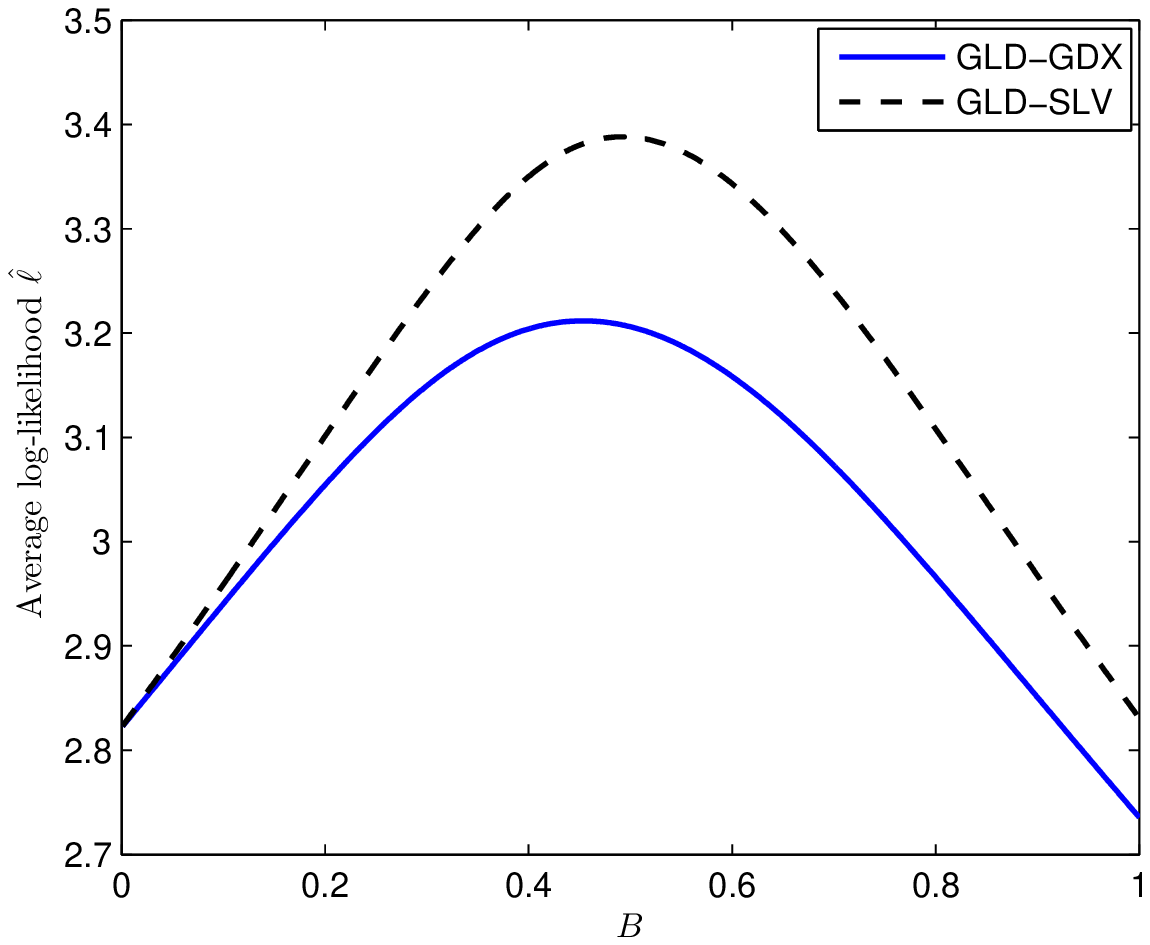}
   \label{fig:loglikeli}
   }
    \subfigure[]{
  \includegraphics[scale=0.58]{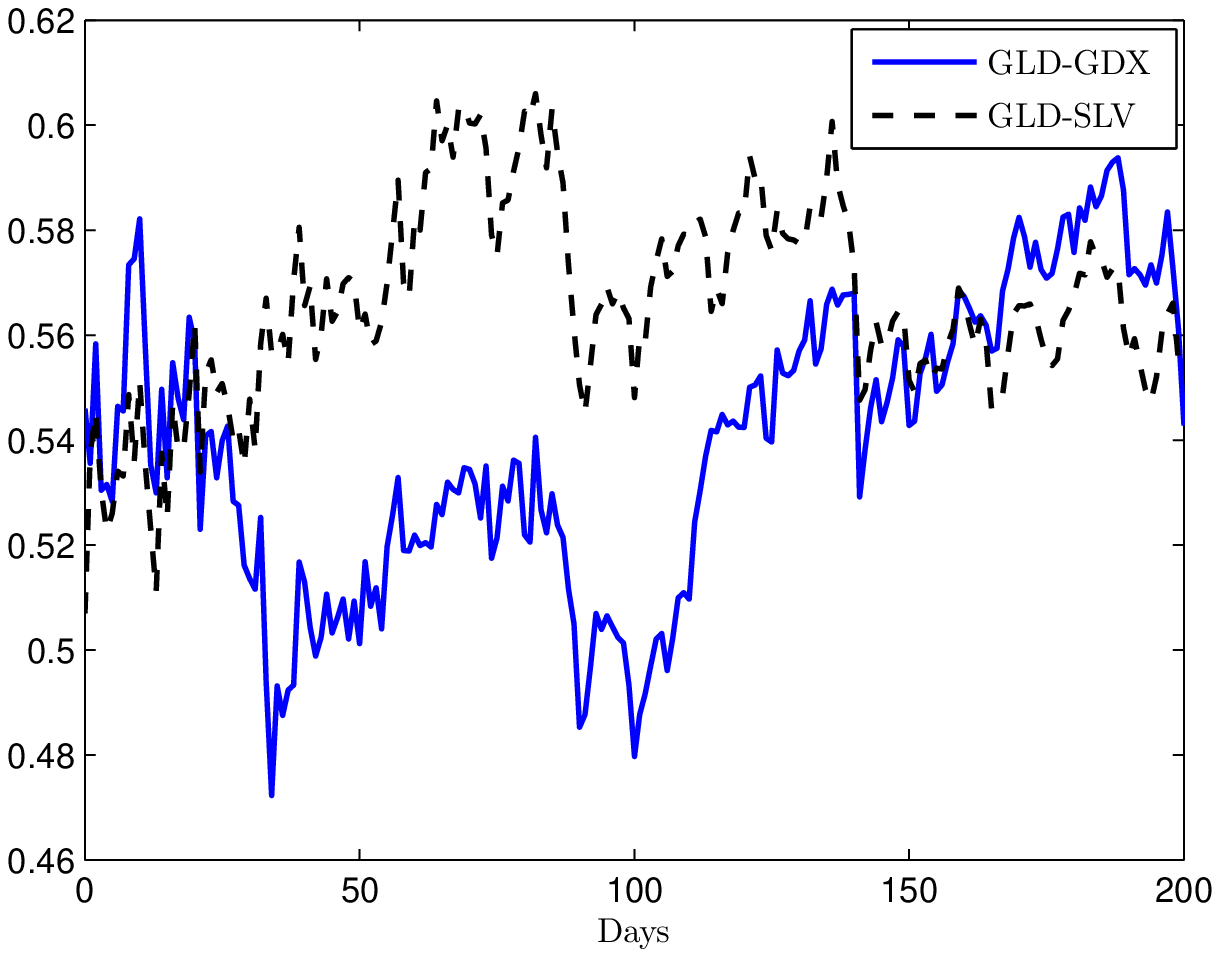}
   \label{fig:path}
   }
 \caption{\small{(a) Average log-likelihood plotted against $B$. (b) Historical price paths with maximum average log-likelihood. The solid line plots the portfolio price with longing $\$1$ GLD and shorting $\$0.454$ GDX, and the dashed line plots the portfolio price with longing $\$1$ GLD and shorting $\$0.493$ SLV.} }\label{fig:MLE}
\end{figure}

\begin{table}[h]
\begin{small}
\begin{center}
    \begin{tabular}{c|c|cccc}
    \hline
  & Price & $\hat{\theta}$  & $\hat{\mu}$ & $\hat{\sigma}$ & $\hat{\ell}$ \\ \hline\hline
   \multirow{2}{*}{GLD-GDX} &     {empirical} & 0.5388 & 16.6677 & 0.1599 & 3.2117\\
    & {simulated} & 0.5425 & 14.3893 & 0.1727 & 3.1304\\ \hline
   \multirow{2}{*}{GLD-SLV} &      {empirical} & 0.5680 &33.4593 & 0.1384 &3.3882\\
  &  {simulated} & 0.5629 & 28.8548  &0.1370 & 3.3898\\ \hline
    \end{tabular}
    \end{center}\end{small}
    \caption{\small{MLE estimates of  OU process parameters using historical prices of GLD, GDX, and SLV from August 2011 to May 2012. The portfolio consists of  $\$1$ in  GLD and  -$\$0.454$ in GDX (resp. -$\$0.493$  in SLV).  For each pair,  the second row (simulated) shows the MLE parameter  estimates  based on a  simulated price path corresponding to the estimated parameters from the first row (empirical).    }} \label{tab:calib}
\end{table}

%


%
\subsection{Optimal Stopping Problem}
Given that a price process or portfolio value evolves according to an OU process, our main objective  is to study the optimal timing to open and subsequently close the position subject to transaction costs.  This leads to the analysis of an optimal double stopping problem.

First, suppose that   the  investor  already has an  existing position whose  value process  $(X_t)_{t\ge 0}$  follows \eqref{XOU}.  If  the position is closed  at some time $\tau$, then the investor will receive the  value  $X_\tau$  and pay a constant transaction cost $c\in \R$. To maximize the expected discounted value,  the investor solves the optimal stopping problem
\begin{align}V(x) = \sup_{\tau \in \setT}\E_x\!\left\{e^{-r \tau}(X_{\tau} - c) \right\}, \label{V1a}\end{align}
where $\setT$ denotes  the set of all $\mathbb{F}$-stopping times, and $r>0$ is the investor's subjective constant discount rate. We have also used the shorthand notation:  $\E_x\{\cdot\}\equiv\E\{\cdot|X_0=x\}$.

From the investor's viewpoint, $V(x)$ represents the expected  liquidation
value associated with  $X$. On the other hand, the current price plus the
transaction cost constitute  the total cost to enter the trade. The investor
can always choose the optimal timing to start the trade, or not to enter at
all. This leads us to analyze the  entry timing inherent in the trading
problem. Precisely, we solve \begin{align}J(x) =  \sup_{\nu \in \setT
}\E_x\!\left\{e^{-\hat{r} \nu} ( V(X_{\nu})  - X_{\nu} - \hat{c})\right\},
\label{J1a}\end{align}with $\hat{r}>0$, $\hat{c}\in\R$.  In other words, the investor seeks
to maximize the expected difference between the value function
$V(X_\nu)$ and the current  $X_\nu$, minus transaction cost $\hat{c}$. The value function $J(x)$
represents the maximum expected value of the investment opportunity in
the price  process $X$, with transaction costs $\hat{c}$ and $c$ incurred, respectively, at entry and exit.  
For  our analysis, the pre-entry and post-entry discount rates,  $\hat{r}$
 and $r$, can be different, as long as $0<\hat{r} \leq r$. Moreover, the
 transaction costs $\hat{c}$ and $c$ can also differ, as long as $c+\hat{c}>0$.  Moreover, since    $\tau = +\infty$ and $\nu = + \infty$ are   candidate stopping times for \eqref{V1a} and \eqref{J1a} respectively, the two value functions $V(x)$ and $J(x)$  are non-negative.

 As extension, we can incorporate a stop-loss level of the pairs trade, that caps the
maximum loss.  In practice, the stop-loss level may be exogenously imposed by
the  manager of a trading desk. In effect, if the price $X$ ever reaches level
$L$ prior to the investor's voluntary liquidation time, then the position  will
be closed immediately. The stop-loss signal is given by the first passage
time
   \[\tau_L := \inf\{ t\ge 0 \,:\, X_t \le L\}.\] Therefore, we determine the entry and  liquidation timing from the constrained optimal stopping problem:
   \begin{align} \label{J1} J_L(x) &=  \sup_{\nu \in \setT }\E_x\left\{e^{-\hat{r} \nu} ( V_L(X_{\nu})  - X_{\nu} - \hat{c})\right\},\\
V_L(x) &= \sup_{\tau \in \setT}\E_x\left\{e^{-r (\tau\wedge \tau_L)}(X_{\tau\wedge \tau_L} - c) \right\}. \label{V1}\end{align}
 Due to the additional timing constraint, the investor may be forced to exit early at the stop-loss level for any given liquidation level. Hence, the stop-loss constraint reduces the   value functions, and precisely  we deduce that  $x-c \leq V_L(x) \leq V(x)$ and $0\leq J_L(x) \leq J(x)$.
 As we will show in Sections \ref{sect-OU} and \ref{sect-stoploss}, the optimal timing   strategies  with and without stop-loss are  quite  different.

 \section{Method of Solution}\label{sect-meth}

In this section, we disucss our method of solution. First, we denote    the infinitesimal generator of the OU process $X$ by  \begin{align}\label{genX}\L = \frac{\sigma^2}{2}\frac{d^2}{d x^2} + \mu(\theta - x)\frac{d}{d x},
\end{align}
and recall the classical  solutions of the differential equation \begin{align}\L u(x)=ru(x),\label{LUXOU}\end{align} for $x\in \R$, are (see e.g.  p.542 of \cite{borodin2002handbook} and Prop. 2.1 of \cite{alili2005representations}):
\begin{align}
F(x) \equiv F(x;r):= \int_0^\infty u^{\frac{r}{\mu}-1} e^{\sqrt{\frac{2\mu}{\sigma^2}}(x-\theta)u-\frac{u^2}{2}} \dx{u},\label{FOU}\\
G(x) \equiv G(x;r):= \int_0^\infty u^{\frac{r}{\mu}-1} e^{\sqrt{\frac{2\mu}{\sigma^2}}(\theta-x)u-\frac{u^2}{2}} \dx{u}.\label{GOU}\end{align}
Direct differentiation yields that   $F'(x) >0$,  $F''(x)>0$,  $G'(x)<0$ and $G''(x)>0$. Hence, we observe that both $F(x)$ and $G(x)$ are strictly positive and convex, and they are, respectively, strictly increasing and decreasing.


Define the first passage time of $X$ to some level $\kappa$ by $\tau_{\kappa} = \inf\{t \geq 0: X_t = \kappa\}$. As is well known, $F$ and $G$ admit the probabilistic expressions (see \cite{ito1965diffusion} and \cite{Rogers2000}):
\begin{align}\label{generalFG1}
\E_x\{e^{-r\tau_\kappa}\} = \begin{cases}
\frac{F(x)}{F(\kappa)} &\, \textrm{ if }\, x\leq \kappa,\\
\frac{G(x)}{G(\kappa)} &\, \textrm{ if }\, x\geq\kappa.
\end{cases}
\end{align}

A key step of our solution  method involves the  transformation  \begin{align}\label{psi2}\psi(x):=  \frac{F}{G}(x).\end{align}
  Starting at  any $x\in \R$, we denote  by  $\tau_a\wedge\tau_b$  the exit time from an  interval $[a,b]$ with  $-\infty \le a \le x \le b \le +\infty$.  With  the reward function $h(x)=x-c$, we compute the corresponding expected discounted reward:
\begin{align}
\E_x\{e^{-r(\tau_a\wedge\tau_b)}h(X_{\tau_a\wedge\tau_b})\} &= h(a)\E_x\{e^{-r\tau_a}\indic{\tau_a<\tau_b}\} + h(b) \E_x\{e^{-r\tau_b}\indic{\tau_a>\tau_b}\}\label{EH1}\\
&= h(a)\frac{F(x)G(b)-F(b)G(x)}{F(a)G(b)-F(b)G(a)} + h(b)\frac{F(a)G(x)-F(x)G(a)}{F(a)G(b)-F(b)G(a)}\label{EH1.5}\\
&= G(x)\left[\frac{h(a)}{G(a)}\frac{\psi(b)-\psi(x)}{\psi(b)-\psi(a)} + \frac{h(b)}{G(b)}\frac{\psi(x)-\psi(a)}{\psi(b)-\psi(a)} \right]\label{EH2}\\
&= G(\psi^{-1}(y))\left[H(y_a)\frac{y_b-y}{y_b-y_a}+H(y_b)\frac{y-y_a}{y_b-y_a} \right], \label{EH3}
\end{align}
where  $y_a=\psi(a)$, $y_b=\psi(b)$, and
\begin{align}\label{generalH}
H(y)  := \begin{cases}
\frac{h}{G}\circ \psi^{-1}(y) &\, \textrm{ if }\, y>0,\\
\lim_{x\to -\infty}\limits\frac{(h(x))^+}{G(x)} &\, \textrm{ if }\, y=0.
\end{cases}
\end{align}
The second equality \eqref{EH1.5} follows from the fact  that $f(x):= \E_x\{e^{-r(\tau_a\wedge\tau_b)}\indic{\tau_a<\tau_b}\}$ is the unique solution to \eqref{LUXOU}  with boundary conditions $f(a)=1$ and $f(b)=0$. Similar reasoning applies to the function $g(x):= \E_x\{e^{-r(\tau_a\wedge\tau_b)}\indic{\tau_a>\tau_b}\}$  with   $g(a)=0$ and $g(b)=1$. The last equality \eqref{EH3} transforms the problem from $x$ coordinate to $y = \psi(x)$  coordinate (see \eqref{psi2}).

 The candidate optimal exit interval $[a^*, b^*]$ is determined by maximizing the expectation in \eqref{EH1}. This is equivalent to maximizing \eqref{EH3} over $y_{a}$ and $y_{b}$  in  the transformed problem. This leads to
\begin{align}
W(y) := \sup_{\{y_a,y_b: y_a\leq y\leq y_b\}} \left[H(y_a)\frac{y_b-y}{y_b-y_a}+H(y_b)\frac{y-y_a}{y_b-y_a}\right].\label{Wy1}
\end{align}
This is the smallest concave majorant of $H$. Applying the definition of $W$ to \eqref{EH3}, we can express the maximal expected discounted reward as
\[G(x)W(\psi(x)) =\sup_{\{a,b: a \leq x \leq b\}} \E_x\{e^{-r(\tau_a\wedge\tau_b)}h(X_{\tau_a\wedge\tau_b})\}. \]

\begin{remark}\label{rmk:OUupper}If  $a=-\infty$, then we have $\tau_{a}=+\infty$  and  $\indic{\tau_a<\tau_b}=0$ a.s. In effect, this removes the  lower exit level, and the corresponding expected discounted reward is
\begin{align}
\E_x\{e^{-r(\tau_a\wedge\tau_b)}h(X_{\tau_a\wedge\tau_b})\} = \E_x\{e^{-r\tau_a}h(X_{\tau_a})\indic{\tau_a<\tau_b}\} +  \E_x\{e^{-r\tau_b}h(X_{\tau_b})\indic{\tau_a>\tau_b}\}= \E_x\{e^{-r\tau_b}h(X_{\tau_b})\}.\notag
\end{align}
 Consequently, by considering interval-type strategies, we also include the class of stopping strategies  of reaching a single upper level $b$  (see Theorem \ref{thm:optLiquOU} below).
\end{remark}

 Next, we prove the optimality of the proposed stopping strategy and provide an expression for the value function.

\begin{theorem}\label{thm:V}
The value function $V(x)$ defined in \eqref{V1a} is given by
\begin{align}\label{generalV}
V(x) = G(x)W(\psi(x)),
\end{align}
where $G$, $\psi$ and $W$ are defined in \eqref{GOU}, \eqref{psi2} and \eqref{Wy1}, respectively.
\end{theorem}

The proof is provided in Appendix {A.1}.   Let us emphasize that the optimal levels $(a^*,  b^*)$ may depend on the initial value  $x$, and can potentially coincide, or  take values $-\infty$ and $+\infty$. As such,  the structure of the stopping and delay regions can potentially  be characterized by {multiple} intervals, leading to \emph{disconnected}  delay regions (see {Theorem \ref{thm:optEntryOUSL}} below).  


We follow the procedure for  Theorem \ref{thm:V} to derive the expression for the  value function $J$ in \eqref{J1a}. First, we   denote $\hat{F}(x)=F(x;\hat{r})$ and $\hat{G}(x)=G(x;\hat{r})$ (see \eqref{FOU}--\eqref{GOU}), with discount rate $\hat{r}$. In addition,   we  define the transformation
\begin{align}\label{generalpsih}
\hat{\psi}(x):= \frac{\hat{F}}{\hat{G}}(x) \quad \text{ and } \quad  \hat{h}(x)= V(x)-x- \hat{c}.
\end{align}
Using these functions, we consider the function analogous to $H$:
\begin{align}\label{generalhatH}
\hat{H}(y) := \begin{cases}
\frac{\hat{h}}{\hat{G}}\circ \hat{\psi}^{-1}(y) &\, \textrm{ if }\, y>0,\\
\lim_{x\to -\infty}\limits\frac{(\hat{h}(x))^+}{\hat{G}(x)} &\, \textrm{ if }\, y=0.
\end{cases}
\end{align}
Following the steps   \eqref{EH1}--\eqref{Wy1} with   $F$, $G$, $\psi$, and $H$  replaced by  $\hat{F}$, $\hat{G}$, $\hat{\psi}$, and $\hat{H}$, respectively,  we write down  the smallest concave majorant  $\hat{W}$  of  $\hat{H}$, namely,
\begin{align*}
\hat{W}(y) := \sup_{\{y_{\hat{a}}, y_{\hat{b}}:y_{\hat{a}}\leq y\leq y_{\hat{b}}\}} \left[\hat{H}(y_{\hat{a}})\frac{y_{\hat{b}}-y}{y_{\hat{b}}-y_{\hat{a}}}+\hat{H}(y_{\hat{b}})\frac{y-y_{\hat{a}}}{y_{\hat{b}}-y_{\hat{a}}}\right].
\end{align*}
From this,  we seek to determine the candidate optimal entry interval $(y_{\hat{a}^*},y_{\hat{b}^*})$ in the $y = \hat{\psi}(x)$ coordinate. Following  the proof of Theorem \ref{thm:V} with the new functions  $\hat{F}$, $\hat{G}$, $\hat{\psi}$,  $\hat{H}$, and $\hat{W}$, the value function of the optimal entry timing problem  admits the expression
\begin{align}\label{generalJ}
J(x) = \hat{G}(x)\hat{W}(\hat{\psi}(x)).
\end{align}

An alternative  way to solve for     $V(x)$  and $J(x)$  is to look for the solutions to  the pair of variational inequalities
\begin{align}
\min\{rV(x)-\L V(x),V(x)-(x-c)\}&=0,\label{VIV}\\
\min\{\hat{r}J(x)-\L J(x),J(x)-(V(x)-x-\hat{c})\}&=0,\label{VIJ}
\end{align}
for $x \in \R$.
With sufficient regularity  conditions, this approach can verify  that the solutions    to the VIs, $V(x)$  and $J(x)$,  indeed correspond to the optimal stopping problems (see, for example, Theorem 10.4.1 of \cite{Oksendal2003}). Nevertheless, this approach does not immediately suggest candidate optimal timing strategies  or  value functions, and typically begins with a conjecture on the structure of the optimal stopping times, followed by verification.  In contrast,   our approach  allows us to directly construct the value functions, at the cost of analyzing the properties of  $H$,  $W$,  $\hat{H}$,  and $\hat{W}$. 


\section{Analytical Results}\label{sect-OU}
We will first study the optimal exit timing in Section \ref{sect-OU-exit}, followed by the optimal entry timing problem in Section \ref{sect-OU-entry}.

\subsection{Optimal Exit Timing}\label{sect-OU-exit}
We now analyze the optimal exit timing problem  \eqref{V1a}. In preparation for the next result,  we 
summarize the crucial properties of $H$.


\begin{lemma}\label{lm:HOU}
The function $H$ is continuous on $[0,+\infty)$, twice differentiable on $(0,+\infty)$ and possesses the following properties:
\begin{enumerate}[(i)]
\item \label{HOU0} $H(0) = 0$, and
\begin{align*}
H(y) \begin{cases}
<0 &\, \textrm{ if }\, y\in (0,\psi(c)),\\
>0 &\, \textrm{ if }\, y \in (\psi(c),+\infty).
\end{cases}
\end{align*}
\item \label{HOU1} Let $x^*$ be the unique solution to $G(x) -
    (x-c)G'(x) =0$. Then,  we have
\begin{align*}
H(y) \textrm{ is strictly}
\begin{cases}
\textrm{decreasing} &\, \textrm{ if }\, y \in (0, \psi(x^*)),\\
\textrm{increasing} &\, \textrm{ if }\, y \in (\psi(x^*),+\infty),
\end{cases}
\end{align*}and $x^*< c\wedge L^*$ with
\begin{align} L^* = \frac{\mu\theta+rc}{\mu+r}.\label{Lstar1}\end{align}

\item \label{HOU2} \begin{align*}H(y) \textrm{ is }
\begin{cases}
\textrm{convex} &\, \textrm{ if }\, y \in (0, \psi(L^*)],\\
\textrm{concave} &\, \textrm{ if }\, y \in [\psi(L^*),+\infty).
\end{cases}
\end{align*}
\end{enumerate}
\end{lemma}

Based on  Lemma \ref{lm:HOU}, we sketch $H$ in Figure \ref{fig:HOU}.  The properties of $H$ are essential  in deriving the value function  and optimal liquidation level, as we show next.

\begin{figure}[th!]
\begin{center}
 {\scalebox{0.35}{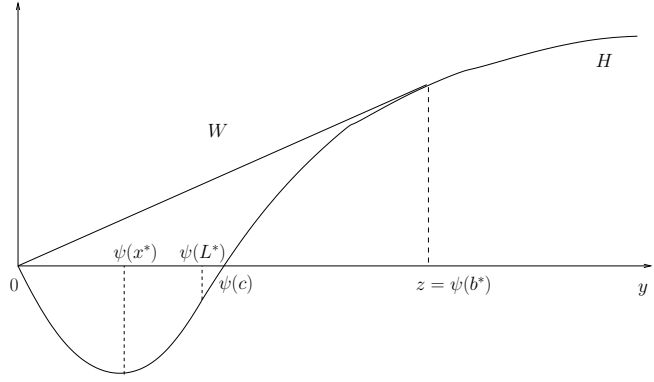}}
\end{center}
\caption{{\small Sketches  of $H$ and $W$. By Lemma \ref{lm:HOU}, $H$ is convex on the left of  $\psi(L^*)$ and concave on the right. The smallest concave majorant  $W$ is a straight line  tangent to $H$ at $z$ on $[0,z)$, and coincides with $H$ on $[z,+\infty)$. }}
\label{fig:HOU}
\end{figure}

\begin{theorem}\label{thm:optLiquOU}
The optimal liquidation problem
\eqref{V1a} admits the solution
\begin{align}V(x) =
\begin{cases}
(b^*-c)\frac{F(x)}{F(b^*)} &\, \textrm{ if }\, x\in (-\infty,b^*),\\
x-c &\, \textrm{ otherwise,}
\end{cases}\label{VOUsol}\end{align}
where the optimal liquidation level $b^*$  is found from the equation
\begin{align}\label{eq:solvebOU}
F(b) = (b-c)F'(b),
\end{align}
and is bounded below by $L^* \vee c$. The corresponding optimal
liquidation time is given by
\begin{align}\tau^* = \inf\{ t\ge 0 \,:\, X_t \ge b^*\}.\label{taustar}
\end{align}
\end{theorem}

\begin{proof} From Lemma \ref{lm:HOU} and the fact that $H'(y)\to 0$ as $y\to +\infty$ (see also Figure \ref{fig:HOU}),  we infer that there exists a unique number $z > \psi(L^*)\vee \psi(c)$ such that \begin{align}\frac{H(z)}{z} = H'(z).\label{eq:HzzLiquOU}\end{align} In turn, the smallest concave majorant is given by
\begin{align}\label{Wy}
W(y) = \begin{cases}
y\frac{H(z)}{z} &\, \textrm{ if }\, y < z,\\
H(y) &\, \textrm{ if }\, y \geq z.
\end{cases}
\end{align}
Substituting  $b^* = \psi^{-1}(z)$ into \eqref{eq:HzzLiquOU}, we have the LHS
\begin{align}\label{fracHz}\frac{H(z)}{z} = \frac{H(\psi(b^*))}{\psi(b^*)} = \frac{b^*-c}{F(b^*)},\end{align}
and the RHS
\begin{align*}
H'(z) &= \frac{G(\psi^{-1}(z))-(\psi^{-1}(z)-c)G'(\psi^{-1}(z))}{F'(\psi^{-1}(z))G(\psi^{-1}(z)) - F(\psi^{-1}(z)) G'(\psi^{-1}(z))}= \frac{G(b^*) - (b^*-c)G'(b^*)}{F'(b^*)G(b^*)-F(b^*)G'(b^*)}.
\end{align*}
Equivalently, we can express  condition \eqref{eq:HzzLiquOU} in terms of $b^*$:
\begin{align*}
\frac{b^*-c}{F(b^*)} = \frac{G(b^*) - (b^*-c)G'(b^*)}{F'(b^*)G(b^*)-F(b^*)G'(b^*)},
\end{align*}
which can be further simplified to
\begin{align*}
F(b^*) = (b^*-c)F'(b^*).
\end{align*}
 Applying \eqref{fracHz} to \eqref{Wy}, we get
\begin{align}\label{WV1}
W(\psi(x)) = \begin{cases}
\psi(x)\frac{H(z)}{z} = \frac{F(x)}{G(x)}\frac{b^*-c}{F(b^*)} &\, \textrm{ if }\, x < b^*,\\
H(\psi(x)) = \frac{x-c}{G(x)} &\, \textrm{ if }\, x \geq b^*.
\end{cases}
\end{align}
In turn, we obtain the value function $V(x)$ by substituting \eqref{WV1} into   \eqref{generalV}.\end{proof}

Next, we examine the dependence of the investor's optimal timing strategy on the transaction cost  $c$.


\begin{proposition}\label{prop:bc}
The value function $V(x)$ of \eqref{V1a} is decreasing in  the transaction cost $c$ for every $x\in \R$, and the   optimal liquidation level $b^*$ is increasing  in  $c$.
\end{proposition}
\begin{proof} For  any $x\in\R$ and $\tau \in \setT$,   the corresponding expected discounted reward,  $\E_x\{e^{-r\tau}(X_{\tau}-c)\} = \E_x\{e^{-r\tau}X_{\tau}\} -c\,\E_x\{e^{-r\tau} \}$, is decreasing in $c$. This implies that $V(x)$ is also decreasing in $c$.  Next, we treat the optimal threshold $b^*(c)$ as a function of $c$,   and differentiate  \eqref{eq:solvebOU} w.r.t. $c$ to get
\begin{align*}
{b^{*}}^\prime\!(c) = \frac{F'(b^*)}{(b^*-c)F''(b^*)} >0 .
\end{align*}
Since $F'(x)>0$, $F''(x)>0$ (see  \eqref{FOU}),  and $b^*>c$ according to  Theorem \ref{thm:optLiquOU}, we conclude that  $b^{*}$ is increasing in $c$.
\end{proof}

In other words, if the transaction cost is high, the investor would tend to liquidate at a higher level, in order to compensate the loss on transaction cost. For other parameters, such as $\mu$ and $\sigma$, the dependence of $b^*$ is generally not monotone.

\subsection{Optimal Entry Timing}\label{sect-OU-entry}
Having solved for the optimal exit timing, we now turn to the optimal entry timing problem. In this case, the value function is
\begin{align}\label{JOUd}
J(x) =  \sup_{\nu \in \setT }\E_x\{e^{-\hat{r} \nu} ( V(X_{\nu})  - X_{\nu} - \hat{c})\},\quad x \in \R,
\end{align}
where $V(x)$   is given by Theorem \ref{thm:optLiquOU}.



To solve for the optimal entry threshold(s), we will need several  properties of $\hat{H}$, as we summarize below.

\begin{lemma}\label{lm:hatHOU}
The function $\hat{H}$ is continuous on $[0,+\infty)$, differentiable on $(0,+\infty)$, and twice differentiable on $(0,\hat{\psi}(b^*)) \cup (\hat{\psi}(b^*),+\infty)$, and possesses the following properties:
\begin{enumerate}[(i)]
\item \label{hatHOU0} $\hat{H} (0) = 0$.
Let $\bar{d}$ denote the unique solution to $\hat{h}(x) =0$, then $\bar{d} < b^*$ and
\begin{align*}
\hat{H}(y) \begin{cases}
>0 &\, \textrm{ if }\, y \in (0, \hat{\psi}(\bar{d})),\\
< 0 &\, \textrm{ if }\, y \in (\hat{\psi}(\bar{d}), +\infty).
\end{cases}
\end{align*}
\item \label{hatHOU1}
$\hat{H}(y)$ is strictly decreasing if $y \in (\hat{\psi}(b^*),+\infty)$.
\item \label{hatHOU2}
Let $\underline{b}$ denote the unique solution to $(\L - \hat{r})\hat{h}(x)=0$, then $\underline{b} < L^*$ and
\begin{align*}
\hat{H}(y) \textrm{ is }
\begin{cases}
\textrm{concave} &\, \textrm{ if }\, y \in (0, \hat{\psi}(\underline{b})),\\
\textrm{convex} &\, \textrm{ if }\, y \in (\hat{\psi}(\underline{b}),+\infty).
\end{cases}
\end{align*}
\end{enumerate}
\end{lemma}

\begin{figure}[th]
\begin{center}

 {\scalebox{0.35}{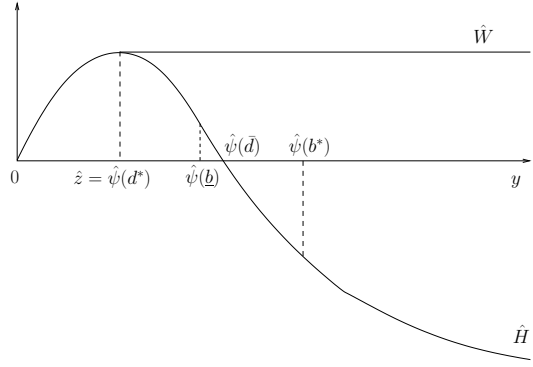}}
\end{center}
\caption{\small{Sketches of $\hat{H}$ and $\hat{W}$. The function $\hat{W}$ coincides with $\hat{H}$ on $[0,\hat{z}]$ and is equal to the constant $\hat{H}(\hat{z})$ on $(\hat{z},+\infty)$. }}
\label{fig:hatHOU}
\end{figure}

In Figure  \ref{fig:hatHOU}, we give a sketch of $\hat{H}$ according to  Lemma \ref{lm:hatHOU}.
 This will be useful for deriving the optimal entry level.
\begin{theorem}\label{thm:optEntryOU}
The optimal entry timing problem
\eqref{J1a} admits the solution
\begin{align}\label{JOUsol}
  J(x) =
\begin{cases}
  V(x)-x-\hat{c} &\, \textrm{ if }\, x \in (-\infty,d^*],\\
  \frac{V(d^*)-d^*-\hat{c}}{\hat{G}(d^*)}\hat{G}(x) &\, \textrm{ if }\, x \in (d^*, +\infty),
\end{cases}
\end{align}
where the optimal entry level $d^*$ is found from the equation
\begin{align}
\label{eq:solvedOU}
\hat{G}(d)(V'(d)  - 1) = \hat{G}'(d)(V(d)-d-\hat{c}).
\end{align}
\end{theorem}

\begin{proof} We look for the value function of  the form: $J(x) = \hat{G}(x)\hat{W}(\hat{\psi}(x))$, where $\hat{W}$ is the the smallest concave majorant of $\hat{H}$.  From Lemma \ref{lm:hatHOU} and  Figure \ref{fig:hatHOU},  we infer that there exists a unique number $\hat{z} < \hat{\psi}(b^*)$ such that
\begin{align}\label{eq:HzEntryOU}
\hat{H}'(\hat{z}) = 0.
\end{align}
This implies that
\begin{align}\label{hatW}
\hat{W}(y) = \begin{cases}
\hat{H}(y) &\, \textrm{ if }\, y \leq \hat{z},\\
\hat{H}(\hat{z}) &\, \textrm{ if }\, y > \hat{z}.
\end{cases}
\end{align}
Substituting $d^* = \hat{\psi}^{-1}(\hat{z})$ into \eqref{eq:HzEntryOU}, we have
\[ \hat{H}'(\hat{z}) = \frac{\hat{G}(d^*)(V'(d^*)  - 1) - \hat{G}'(d^*)(V(d^*)-d^*-\hat{c})}
{\hat{F}'(d^*)\hat{G}(d^*)-\hat{F}(d^*)\hat{G}'(d^*)}=0,\]
which is equivalent to condition \eqref{eq:solvedOU}. Furthermore, using \eqref{generalpsih} and \eqref{generalhatH}, we get
\begin{align}\label{hatHz} \hat{H}(\hat{z}) = \frac{V(d^*)-d^*-\hat{c}}{\hat{G}(d^*)}.\end{align}
To conclude, we   substitute   $\hat{H}(\hat{z})$ of \eqref{hatHz} and  $\hat{H}(y)$ of  \eqref{generalhatH} into  $\hat{W}$ of  \eqref{hatW}, which by \eqref{generalJ} yields  the value function $J(x)$ in \eqref{JOUsol}.
\end{proof}



With the analytic solutions for $V$ and $J$, we can verify  by direct substitution that   $V(x)$ in  \eqref{VOUsol} and $J(x)$  in \eqref{JOUsol} satisfy both \eqref{VIV} and \eqref{VIJ}.

Since the optimal entry timing problem is  nested with another optimal
stopping problem, the parameter dependence of the optimal entry level is
complicated. Below, we illustrate the impact of transaction cost.

\begin{proposition}\label{prop:dstar}
The optimal entry level $d^*$ of
\eqref{J1a} is  decreasing  in  the  transaction cost $\hat{c}$.
\end{proposition}
\begin{proof}
 Considering the optimal entry level  $d^*$ as a function of $\hat{c}$, we differentiate   \eqref{eq:solvedOU} w.r.t. $\hat{c}$ to get
\begin{align}
{d^{*}}^\prime\!(\hat{c}) = \frac{-\hat{G}'(d^*)}{\hat{G}(d^*)}[V''(d^*)-\frac{V(d^*)-d^*-\hat{c}}{\hat{G}(d^*)}\hat{G}''(d^*)]^{-1}.\label{dc}
\end{align}
Since $\hat{G}(d^*) >0$ and $\hat{G}'(d^*)<0$, the sign of {$d^{*'}\!(\hat{c})$} is determined by $V''(d^*)-\frac{V(d^*)-d^*-\hat{c}}{\hat{G}(d^*)}\hat{G}''(d^*)$. Denote $\hat{f}(x) = \frac{V(d^*)-d^*-\hat{c}}{\hat{G}(d^*)}\hat{G}(x)$. Recall that $\hat{h}(x) = V(x)-x-\hat{c}$,
\begin{align*}
J(x)=\begin{cases}
\hat{h}(x) &\, \textrm{ if }\, x \in (-\infty, d^*],\\
\hat{f}(x)>\hat{h}(x) &\, \textrm{ if }\, x \in (d^*,+\infty),
\end{cases}
\end{align*}
and $ \hat{f}(x)$ smooth pastes  $\hat{h}(x)$ at $d^*$.  Since both $\hat{h}(x)$ and $ \hat{f}(x)$ are positive decreasing convex functions, it follows that $\hat{h}''(d^*)\leq \hat{f}''(d^*)$. Observing that  $\hat{h}''(d^*) = V''(d^*)$ and $\hat{f}''(d^*)= \frac{V(d^*)-d^*-\hat{c}}{\hat{G}(d^*)}\hat{G}''(d^*)$, we have $V''(d^*)-\frac{V(d^*)-d^*-\hat{c}}{\hat{G}(d^*)}\hat{G}''(d^*) \leq 0$. Applying this to \eqref{dc}, we conclude that  ${d^{*}}^\prime\!(\hat{c}) \leq 0$.\end{proof}

We end this section with a special example in the OU model with  no mean
reversion. \begin{remark}\label{rmk:BM}  If we set  $\mu =0$ in  \eqref{XOU},
with $r$ and $\hat{r}$ fixed, it follows  that   $X$
reduces to a Brownian motion: $X_{t}=  \sigma B_{t}$,  $t\ge 0$. In this case, the
optimal liquidation level $b^*$ for problem \eqref{V1a} is \begin{align*}
b^* = c+\frac{\sigma}{\sqrt{2r}}, \end{align*} and the optimal entry
level $d^*$ for problem \eqref{J1a} is the root  to the equation
\begin{align*} \left(1+\sqrt{\frac{\hat{r}}{r}}
\right)e^{\frac{\sqrt{2r}}{\sigma}(d-c-\frac{\sigma}{\sqrt{2r}})} =
\frac{\sqrt{2\hat{r}}}{\sigma}(d+\hat{c})+1, \quad d \in (-\infty, b^*). \end{align*} \end{remark}

\section{Incorporating Stop-Loss Exit}\label{sect-stoploss}Now we consider the optimal entry and exit problems with a stop-loss constraint. For convenience, we restate the value functions  from \eqref{J1} and \eqref{V1}:
\begin{align}
J_L(x) &=  \sup_{\nu \in \setT }\E_x\left\{e^{-\hat{r} \nu} ( V_L(X_{\nu})  - X_{\nu} - \hat{c})\right\}, \label{J2}\\
V_L(x) &= \sup_{\tau \in \setT}\E_x\left\{e^{-r (\tau\wedge \tau_L)}(X_{\tau\wedge \tau_L} - c) \right\}. \label{V2}
\end{align} After solving for the optimal timing strategies, we will also examine the dependence of  the optimal liquidation threshold   on the stop-loss level $L$.

\subsection{Optimal Exit Timing}
We first give an analytic solution to the optimal exit timing problem.
\begin{theorem}\label{thm:optLiquOUSL}
The optimal liquidation problem
\eqref{V2} with stop-loss level $L$ admits the solution
\begin{align}\label{VOUSLsol}V_L(x) =
\begin{cases}
CF(x)+DG(x) &\, \textrm{ if }\, x\in (L,b_L^*),\\
x-c &\, \textrm{ otherwise},
\end{cases}\end{align}
where
\begin{align}\label{eq:CDOUSL}
C = \frac{(b_L^*-c)G(L) - (L-c)G(b_L^*) }{F(b_L^*)G(L) -F(L)G(b_L^*) },\quad
D = \frac{(L-c)F(b_L^*)-(b_L^*-c)F(L)}{F(b_L^*)G(L) -F(L)G(b_L^*)}.
\end{align}
The optimal liquidation level  $b_L^*$ is found from the equation
\begin{align}\label{eq:solvebOUSL}
[(L-c)G(b)-(b-c)G(L)]F'(b) &+ [(b-c)F(L)-(L-c)F(b)]G'(b) = G(b)F(L)-G(L)F(b).
\end{align}
\end{theorem}

\begin{figure}[th]
\begin{center}
 \scalebox{0.34}{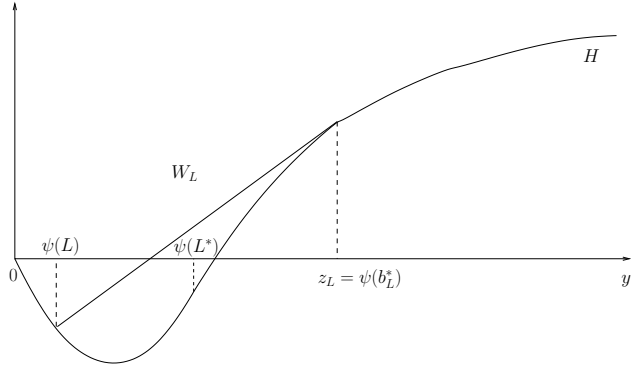}
\end{center}
\caption{\small{Sketch of $W_L$. On $[0,\psi(L)]\cup [z_L,+\infty)$, $W_L$ coincides with $H$, and over $(\psi(L),z_L)$, $W_L$ is a straight line  tangent to $H$ at $z_L$ .}} \label{figWL}
\end{figure}

\begin{proof}Due to the stop-loss level $L$,  we consider  the smallest
concave majorant of $H(y)$, denoted by $W_L(y)$,  over the restricted
domain $[\psi(L),+\infty)$ and  set $W_L(y) = H(y)$ for $y \in [0,\psi(L)]$.

From   Lemma \ref{lm:HOU} and Figure \ref{figWL}, we   see that $H(y)$
is convex over $(0, \psi(L^*)]$ and concave in $[\psi(L^*),+\infty)$.   If $L
\geq L^*$, then $H(y)$ is concave over $[\psi(L),+\infty)$, which implies
that $W_L(y) = H(y)$ for $y\ge 0$, and thus $V_L(x) = x-c$ for $x\in \R$.
On the other hand, if $L < L^*$, then $H(y)$ is convex on
$[\psi(L),\psi(L^*)]$, and concave strictly increasing on $[\psi(L^*),+\infty)$. There
exists a unique number $z_L > \psi(L^*)$ such that
\begin{align}\label{eq:HzLiquOUSL}
\frac{H(z_L) - H(\psi(L))}{z_L-\psi(L)} = H'(z_L).
\end{align}
In turn, the smallest concave majorant admits the form:
\begin{align}\label{eq:WLy}
W_L(y) = \begin{cases}
H(\psi(L)) + (y-\psi(L))H'(z_L) &\, \textrm{ if }\, y \in (\psi(L), z_L),\\
H(y) &\, \textrm{ otherwise. }
\end{cases}
\end{align}

Substituting $b_L^* = \psi^{-1}(z_L)$ into \eqref{eq:HzLiquOUSL}, we
have from the LHS
\begin{align*}
\frac{H(z_L) - H(\psi(L))}{z_L-\psi(L)} = \frac{H(\psi(b_L^*))- H(\psi(L))}{\psi(b_L^*)-\psi(L)}
= \frac{\frac{b_L^*-c}{G(b_L^*)} -\frac{L-c}{G(L)} }{ \frac{F(b_L^*)}{G(b_L^*)}- \frac{F(L)}{G(L)}}
= C,
\end{align*}
and the RHS
\begin{align*}
H'(z_L) &= \frac{G(\psi^{-1}(z_L))-(\psi^{-1}(z_L)-c)G'(\psi^{-1}(z_L))}{F'(\psi^{-1}(z_L))G(\psi^{-1}(z_L)) - F(\psi^{-1}(z_L)) G'(\psi^{-1}(z_L))}\\
&= \frac{G(b^*_L) - (b^*-c)G'(b^*_L)}{F'(b^*_L)G(b^*_L)-F(b^*_L)G'(b^*_L)}.
\end{align*}
Therefore, we can equivalently   express  \eqref{eq:HzLiquOUSL} in terms of $b_L^*$:
\begin{align*}
\frac{(b_L^*-c)G(L) - (L-c)G(b_L^*) }{F(b_L^*)G(L) -F(L)G(b_L^*) }  = \frac{G(b_L^*) - (b_L^*-c)G'(b_L^*)}{F'(b_L^*)G(b_L^*)-F(b_L^*)G'(b_L^*)},
\end{align*}
which by rearrangement immediately simplifies to   \eqref{eq:solvebOUSL}.

Furthermore, for $x \in (L, b_L^*)$, $H'(z_L) = C$ implies that
\begin{align*}
W_L(\psi(x)) = H(\psi(L))+ (\psi(x) -\psi(L))C.
\end{align*}
Substituting this to $V_L(x) = G(x)W_L(\psi(x))$, the value function becomes
\begin{align*}
V_L(x) &= G(x) \big[H(\psi(L))+ (\psi(x) -\psi(L))C\big] = CF(x) + G(x)\big[H(\psi(L)) - \psi(L)C\big],
\end{align*}
which resembles  \eqref{VOUSLsol} after the observation that
\begin{align*}
H(\psi(L)) - \psi(L)C &= \frac{L-c}{G(L)} - \frac{F(L)}{G(L)} \frac{(b_L^*-c)G(L) - (L-c)G(b_L^*) }{F(b_L^*)G(L) -F(L)G(b_L^*) }\\
&= \frac{(L-c)F(b_L^*)-(b_L^*-c)F(L)}{F(b_L^*)G(L) -F(L)G(b_L^*)} = D.
\end{align*}\end{proof}

We can  interpret the investor's timing strategy in terms of three price intervals, namely,  the liquidation
region $[b^*_L, +\infty)$, the delay region $(L, b^*_L)$, and the
stop-loss region $(-\infty, L]$. In both liquidation and stop-loss regions, the
value function $V_L(x) = x-c$, and therefore, the investor will immediately
close out the position.    From the proof of Theorem
\ref{thm:optLiquOUSL}, if $L \geq L^* = \frac{\mu \theta + r c}{\mu +
r}$ (see \eqref{Lstar1}), then $V_L(x) = x-c$, $\forall x\in \R$. In other
words, if the stop-loss level is too high, then the delay  region
completely disappears, and the investor will liquidate immediately for every
initial value $x \in \R$.

\begin{corollary}\label{thm:optb} If   $L <  L^*$, then there
exists a unique solution $b_L^* \in (L^*, +\infty)$ that solves
\eqref{eq:solvebOUSL}. If $L  \ge  L^*$, then  $V_L(x) = x-c$, for $x\in
\R$.  \end{corollary}


The direct effect of a stop-loss exit constraint is  forced liquidation whenever the price process reaches $L$ before the upper liquidation level $b^*_L$. Interestingly, there is an additional indirect effect:    a higher stop-loss level will induce the investor to \emph{voluntarily} liquidate  earlier at a lower take-profit level.

\begin{proposition}\label{prop:bL}
The optimal liquidation level $b_L^*$ of
\eqref{V2}   strictly decreases as the stop-loss level $L$ increases.
\end{proposition}

\begin{proof}
Recall that $z_L = \psi(b_L^*)$ and $\psi$ is a strictly increasing function. Therefore,  it is sufficient to  show that
$z_L$ strictly decreases as   $\tilde{L} := \psi(L)$  increases. As such, we denote  $z_L(\tilde{L})$ to highlight  its dependence on $\tilde{L}$. Differentiating \eqref{eq:HzLiquOUSL} w.r.t. $\tilde{L}$ gives


\begin{align}
z'_L(\tilde{L}) = \frac{H'(z_L)-H'(\tilde{L})}{H''(z_L)(z_L-\tilde{L})}.\label{dzl}
\end{align}
It follows from the definitions of $W_L$ and $z_L$ that
$H'(z_L)>H'(\tilde{L})$ and $z_L>\tilde{L}$. Also, we have $H''(z)<0$
since $H$ is concave at $z_L$. Applying these to \eqref{dzl}, we conclude
that $z'_L(\tilde{L})<0$.
\end{proof}

Figure \ref{fig:admL} illustrates  the optimal exit price level   $b_L^*$ as a
function of the stop-loss levels $L$,  for different long-run means $\theta$.
When $b^*_L$ is strictly  greater than $L$ (on the left of the straight line),
the delay  region is non-empty. As $L$ increases, $b^*_L$ strictly decreases
and the two  meet at $L^*$ (on the straight line), and the delay 
region vanishes.

\begin{figure}[t!] \centering
\includegraphics[scale=0.7]{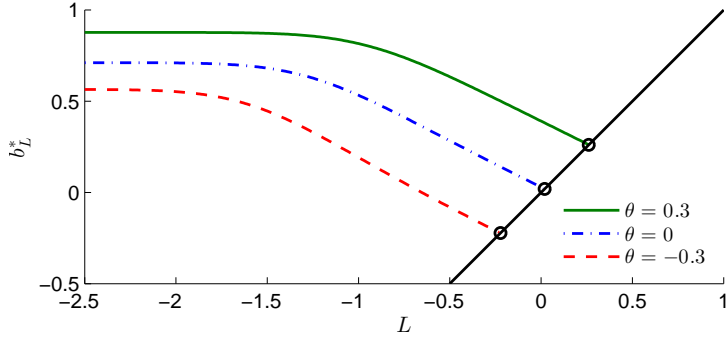} \caption{\small{The
optimal exit threshold $b_L^*$ is strictly  decreasing with respect to the stop-loss
level $L$.  The straight  line is where  $b_L^* = L$, and each
 of the three circles locates the critical stop-loss level  $L^*$.  }} \label{fig:admL}
\end{figure}

Also, there is  an interesting connection  between cases with different
long-run means  and transaction costs. To this end, let us denote the value
function by $V_L(x;\theta,c)$ to highlight the dependence on $\theta$ and
$c$, and  the corresponding optimal liquidation level by $b_L^*(\theta,c)$.
We find that, for any $\theta_1,\theta_2 \in \R$, $c_1, c_2 > 0$, $L_1 \le
 \frac{\mu\theta_1+rc_1}{\mu+r}$, and $L_2 \le \frac{\mu\theta_2+rc_2}{\mu+r}$,   the associated value functions and  optimal
liquidation levels satisfy the relationships: \begin{align}
V_{L_1}(x+\theta_1;\theta_1,c_1)
&=V_{L_2}(x+\theta_2;\theta_2,c_2),\label{Vequiv}\\
b_{L_1}^*(\theta_1,c_1) - \theta_1 &=
b_{L_2}^*(\theta_2,c_2) - \theta_2,\label{bequiv}
\end{align}as long as $\theta_1- \theta_2=
c_1-c_2 = L_1- L_2.$    These results \eqref{Vequiv} and
\eqref{bequiv} also hold in the case without stop-loss.

\subsection{Optimal Entry Timing}

 We now discuss the optimal entry timing  problem $J_L(x)$ defined in \eqref{J2}.   Since  $\sup_{x\in\R}( V_L(x)  -x - \hat{c})\leq 0$  implies that   $J_L(x) = 0$ for $x \in \R$,  we can   focus on the  case with
\begin{align}\label{generalassume} \sup_{x\in\R}( V_L(x)  -x - \hat{c})> 0,
\end{align}and look for  non-trivial optimal  timing strategies.

  Associated with reward function $\hat{h}_L(x) := V_L(x)-x-\hat{c}$ from entering the market, we define the function $\hat{H}_L$ as in \eqref{generalH} whose properties are summarized in the following lemma.




\begin{lemma}\label{lm:hatHOUSL}
The function $\hat{H}_L$ is continuous on $[0,+\infty)$, differentiable on $(0,\hat{\psi}(L)) \cup (\hat{\psi}(L), +\infty)$, twice differentiable on $(0,\hat{\psi}(L)) \cup (\hat{\psi}(L),\hat{\psi}(b^*_L)) \cup (\hat{\psi}(b^*_L),+\infty)$, and possesses the following properties:
\begin{enumerate}[(i)]
\item \label{hatHOUSL0}
$\hat{H}_L(0)=0$. $\hat{H}_L(y)<0$ for $y \in (0, \hat{\psi}(L)]\cup [\hat{\psi}(b_L^*),+\infty)$.

\item  \label{hatHOUSL1}
$\hat{H}_L(y)$ is strictly decreasing for $y \in (0, \hat{\psi}(L))\cup (\hat{\psi}(b_L^*),+\infty)$.

\item \label{hatHOUSL2} There exists some constant $\bar{d}_L \in (L,
    b_L^*)$ such that  $(\L - \hat{r})\hat{h}_L(\bar{d}_L)=0$, and
\begin{align*}
\hat{H}_L(y) \textrm{ is } \begin{cases}
\textrm{convex} &\, \textrm{ if }\, y\in (0, \hat{\psi}(L))\cup (\hat{\psi}(\bar{d}_L),+\infty),\\
\textrm{concave} &\, \textrm{ if }\, y \in (\hat{\psi}(L), \hat{\psi}(\bar{d}_L)).
\end{cases}
\end{align*}
In addition, $\hat{z}_1 \in (\hat{\psi}(L), \hat{\psi}(\bar{d}_L))$, where $\hat{z}_1 := \argmax_{y\in [0,+\infty)} \hat{H}_L(y)$.

\end{enumerate}
\end{lemma}



\begin{theorem}\label{thm:optEntryOUSL}
The optimal entry timing problem
\eqref{J2} admits the solution
\begin{align}\label{JsoluOUSL}
  J_L(x) =
\begin{cases}
  P\hat{F}(x) &\, \textrm{ if }\, x \in (-\infty,a_L^*),\\
  V_L(x)-x-\hat{c} &\, \textrm{ if }\, x \in [a_L^*,d_L^*],\\
  Q\hat{G}(x) &\, \textrm{ if }\, x \in (d_L^*,+\infty),
\end{cases}
\end{align}
where
\begin{align}
P = \frac{V_L(a_L^*)-a_L^*-\hat{c}}{\hat{F}(a_L^*)}, \quad
Q &= \frac{V_L(d_L^*)-d_L^*-\hat{c}}{\hat{G}(d_L^*)}.
\end{align} The optimal entry time is given by
\begin{align}\label{nu_OUSL}
\nu_{a_L^*,d_L^*} = \inf
\{  t\ge 0 \,:\, X_t \in [a_L^*, d_L^*]\},
\end{align}
where the critical levels $a_L^*$
and $d_L^*$ satisfy, respectively,
\begin{align}\label{eq:solveaOUSL} \hat{F}(a)(V_L'(a)  - 1) =
\hat{F}'(a)(V_L(a)-a-\hat{c}), \end{align}
and
\begin{align}\label{eq:solvedOUSL}
\hat{G}(d)(V_L'(d)  - 1) = \hat{G}'(d)(V_L(d)-d-\hat{c}).
\end{align}
\end{theorem}

\begin{figure}[th]
\begin{center}
 %
 {\scalebox{0.35}{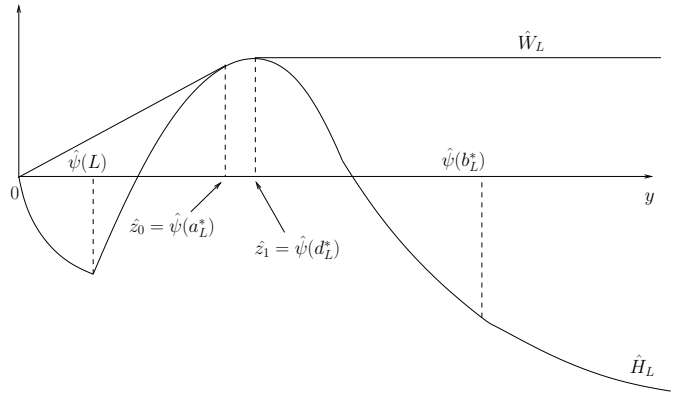}}
\caption{\small{Sketches of $\hat{H}_L$ and $\hat{W}_L$. $\hat{W}_L$ is a straight line tangent to $\hat{H}_L$ at $\hat{z}_0$ on $[0,\hat{z}_0)$, coincides with $\hat{H}_L$ on $[\hat{z}_0,\hat{z}_1]$, and is equal to the constant $\hat{H}_L(\hat{z}_1)$ on $(\hat{z}_1,+\infty)$. Note that $\hat{H}_L$ is not differentiable at $\hat{\psi}(L)$. }}
\label{fig:hatHOUSL}
\end{center}
\end{figure}

\begin{proof}
 We look for the value function of the form: $J_L(x) = \hat{G}(x)\hat{W}_L(\hat{\psi}(x)),$ where $\hat{W}_L$ is the smallest non-negative concave majorant of $\hat{H}_L$. From Lemma \ref{lm:hatHOUSL} and the sketch of $\hat{H}_L$ in Figure \ref{fig:hatHOUSL},
the maximizer of $\hat{H}_L$, $\hat{z}_1$,  satisfies
\begin{align}\label{eq:Hz1EntryOUSL}
\hat{H}_L'(\hat{z}_1)=0.
\end{align}
Also there exists a unique number $\hat{z}_0 \in (\hat{\psi}(L),\hat{z}_1)$ such that
\begin{align}\label{eq:Hz0EntryOUSL}
\frac{\hat{H}_L(\hat{z}_0)}{\hat{z}_0} =\hat{H}_L'(\hat{z}_0).
\end{align}
In turn, the smallest non-negative concave majorant admits the form:
\begin{align*}
\hat{W}_L(y) = \begin{cases}
y\hat{H}_L'(\hat{z}_0) &\, \textrm{ if }\, y\in [0,\hat{z}_0),\\
\hat{H}_L(y) &\, \textrm{ if }\, y\in [\hat{z}_0, \hat{z}_1],\\
\hat{H}_L(\hat{z}_1) &\, \textrm{ if }\, y\in (\hat{z}_1, +\infty).
\end{cases}
\end{align*}
Substituting $a_L^* = \hat{\psi}^{-1}(\hat{z}_0)$ into  \eqref{eq:Hz0EntryOUSL}, we have
\begin{align*}
\frac{\hat{H}_L(\hat{z}_0)}{\hat{z}_0} &= \frac{V_L(a_L^*)-a_L^*-\hat{c}}{\hat{F}(a_L^*)},\\
\hat{H}_L'(\hat{z}_0)&= \frac{\hat{G}(a_L^*)(V_L'(a_L^*)  - 1) - \hat{G}'(a_L^*)(V_L(a_L^*)-a_L^*-\hat{c})}
{\hat{F}'(a_L^*)\hat{G}(a_L^*)-\hat{F}(a_L^*)\hat{G}'(a_L^*)}.
\end{align*}
Equivalently, we can express condition \eqref{eq:Hz0EntryOUSL} in terms of $a_L^*$:
\begin{align*}
\frac{V_L(a_L^*)-a_L^*-\hat{c}}{\hat{F}(a_L^*)} = \frac{\hat{G}(a_L^*)(V_L'(a_L^*)  - 1) - \hat{G}'(a_L^*)(V_L(a_L^*)-a_L^*-\hat{c})}
{\hat{F}'(a_L^*)\hat{G}(a_L^*)-\hat{F}(a_L^*)\hat{G}'(a_L^*)}.
\end{align*}
Simplifying this shows that $a_L^*$ solves \eqref{eq:solveaOUSL}. Also, we can express $\hat{H}_L'(\hat{z}_0)$ in terms of $a_L^*$:
\begin{align*}
\hat{H}_L'(\hat{z}_0)=\frac{\hat{H}_L(\hat{z}_0)}{\hat{z}_0} = \frac{V_L(a_L^*)-a_L^*-\hat{c}}{\hat{F}(a_L^*)} = P.
\end{align*}
In addition, substituting $d_L^* = \hat{\psi}^{-1}(\hat{z}_1)$ into \eqref{eq:Hz1EntryOUSL}, we have
\begin{align*}
\hat{H}_L'(\hat{z}_1) = \frac{\hat{G}(d_L^*)(V_L'(d_L^*)  - 1) - \hat{G}'(d_L^*)(V_L(d_L^*)-d_L^*-\hat{c})}
{\hat{F}'(d_L^*)\hat{G}(d_L^*)-\hat{F}(d_L^*)\hat{G}'(d_L^*)}=0,
\end{align*}
which, after a straightforward  simplification,  is identical to  \eqref{eq:solvedOUSL}.
Also,  $\hat{H}_L(\hat{z}_1)$ can be written as
\begin{align*}
\hat{H}_L(\hat{z}_1) = \frac{V_L(d_L^*)-d_L^*-\hat{c}}{\hat{G}(d_L^*)}=Q.
\end{align*}
Substituting these to $J_L(x) = \hat{G}(x)\hat{W}_L(\hat{\psi}(x))$, we arrive at \eqref{JsoluOUSL}.
\end{proof}

 Theorem \ref{thm:optEntryOUSL} reveals that the optimal entry region is characterized by a price interval $[a^*_L, d^*_L]$ strictly  above  the
 stop-loss level $L$ and strictly below the optimal exit level $b^*_L$.  In particular, if the current asset price is between $L$ and $a_L^*$, then it is optimal for the investor to wait even though the price is low. This is intuitive because  if the entry price is too close to $L$, then  the investor is very likely to be forced to exit at a loss afterwards.  As a consequence,  the investor's delay  region, where she would wait to enter the market,  is  disconnected.

Figure  \ref{fig:SL_paths_3} illustrates  two simulated paths and the associated exercise times. We have chosen $L$ to be 2 standard deviations below the long-run mean $\theta$, with other parameters from our pairs trading example.
 By Theorem \ref{thm:optEntryOUSL}, the investor will enter the market at $\nu_{a_L^*,d_L^*}$ (see \eqref{nu_OUSL}). Since both paths start with $X_0>d_L^*$, the investor waits to enter until the OU path reaches $d_L^*$ from above, as indicated by  $\nu^*_d$ in panels (a) and (b).    After entry, Figure \ref{fig:SL_db_3} describes  the scenario where  the investor exits voluntarily  at the optimal level $b^*_L$, whereas  in Figure \ref{fig:SL_dL_3} the investor is forced to exit at the stop-loss level $L$.  These optimal levels are calculated from Theorem \ref{thm:optEntryOUSL} and Theorem \ref{thm:optLiquOUSL} based on the given estimated parameters.

\begin{figure}[h]
 \centering
 \subfigure[]{
  \includegraphics[scale=0.6]{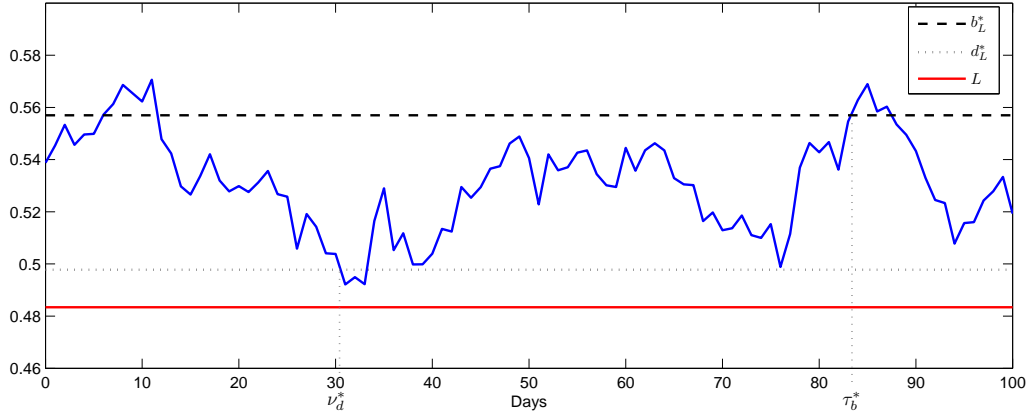}
   \label{fig:SL_db_3}
   }\\
    \subfigure[]{
  \includegraphics[scale=0.6]{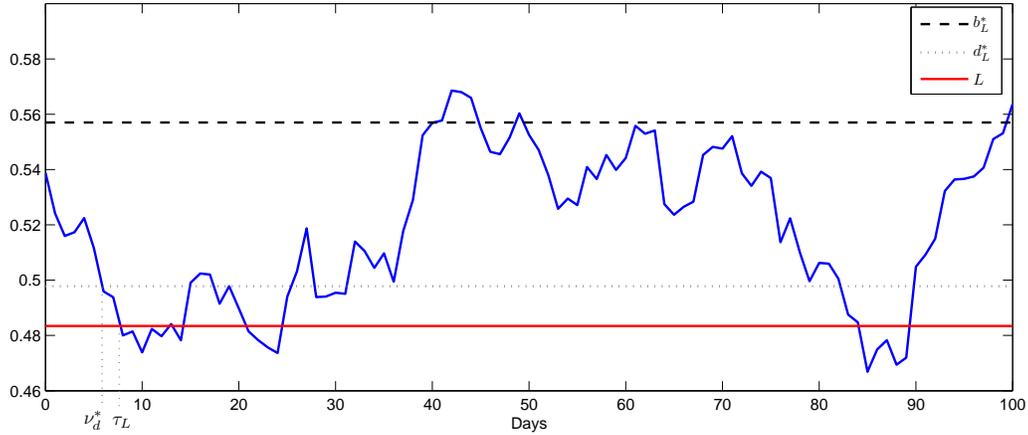}
   \label{fig:SL_dL_3}
   }
 \caption{\small{Simulated OU paths and exercise times.  (a) The investor enters at $\nu^*_d = \inf
\{  t\ge 0 \,:\, X_t \leq d_L^*\}$ with $d_L^*=0.4978$, and exit at  $\tau^*_b = \inf
\{  t\ge \nu^*_d \,:\, X_t \ge b_L^*\}$ with $b_L^* = 0.5570$.  (b)  The investor enters at $\nu^*_d= \inf
\{  t\ge 0 \,:\, X_t \leq d_L^*\}$ but exits at stop-loss level $L=0.4834$.   Parameters:   $\theta=0.5388$, $\mu = 16.6677$, $\sigma = 0.1599$, $ r=\hat{r}=0.05$, and $c=\hat{c}=0.05$.} }\label{fig:SL_paths_3} 
\end{figure}

\begin{remark}
We remark that the optimal levels $a_L^*$, $d_L^*$ and $b_L^*$ are outputs of the models, depending on the parameters $(\mu, \theta, \sigma)$ and the choice of stop-loss level $L$.  Recall that our model parameters are estimated based on the likelihood maximizing portfolio discussed in Section \ref{sect-pairs}. Other estimation methodologies and price data can be used, and may lead to different portfolio strategies $(\alpha, \beta)$ and estimated parameters values $(\mu, \theta, \sigma)$. In turn, the resulting optimal entry and exit thresholds may also change accordingly.
\end{remark}


\subsection{Relative Stop-Loss Exit}\label{sect-rela-sl}
 For some investors, it may be more desirable
  to set the stop-loss contingent on  the entry level.
  In other words, if the value of $X$ at the entry time  is $x$, then  the investor would
   assign a lower stop-loss level $x-\ell$, for some constant $\ell >0$. Therefore, the investor faces the optimal entry timing problem
\begin{align}\label{JrelaSL}
\mathcal{J}_\ell(x)= \sup_{\nu \in \setT }\E_x\left\{e^{-\hat{r} \nu} ( \mathcal{V}_\ell(X_{\nu})  - X_{\nu} - \hat{c})\right\},
\end{align}
where $\mathcal{V}_\ell(x) := V_{x-\ell}(x)$ (see \eqref{V2}) is the
optimal exit timing problem with stop-loss level $x-\ell$. The dependence of
$V_{x-\ell}(x)$ on $x$ is  significantly more complicated than $V(x)$ or $V_L(x)$,
making  the problem much less tractable. In Figure \ref{fig:rela_SL}, we
illustrate numerically the optimal timing strategies. The investor will still
enter at a lower level $d^*$. After entry, the investor will wait to exit at
either the stop-loss level $d^*-\ell$ or an upper level $b^*$.

\begin{figure}[ht]
\begin{center}\includegraphics[width=3.38in]{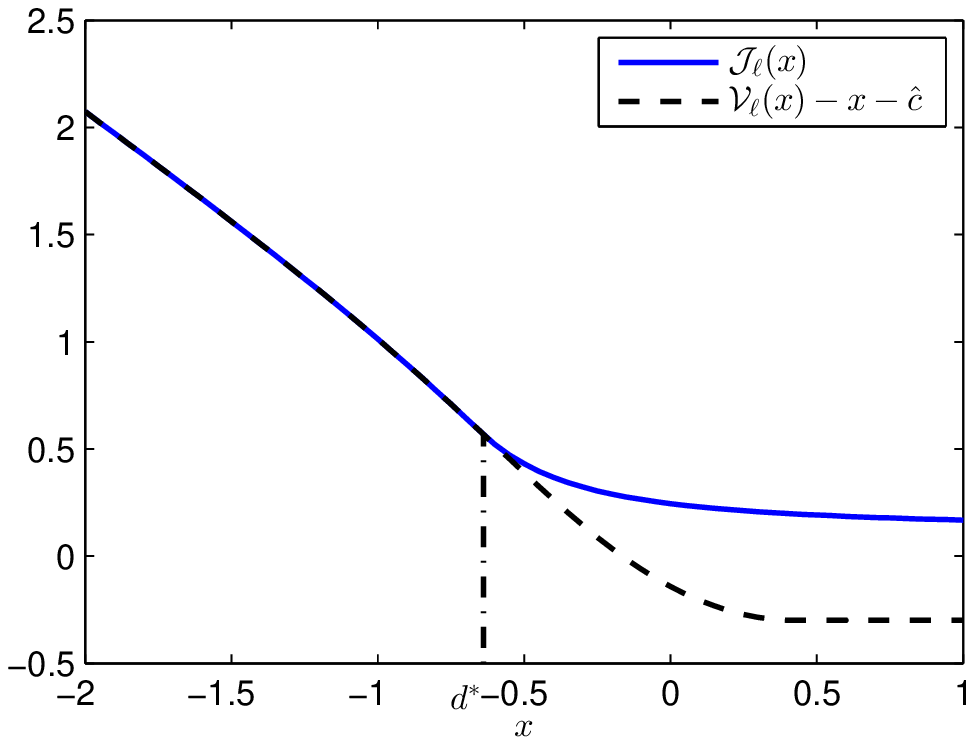}
\includegraphics[width=3.4in]{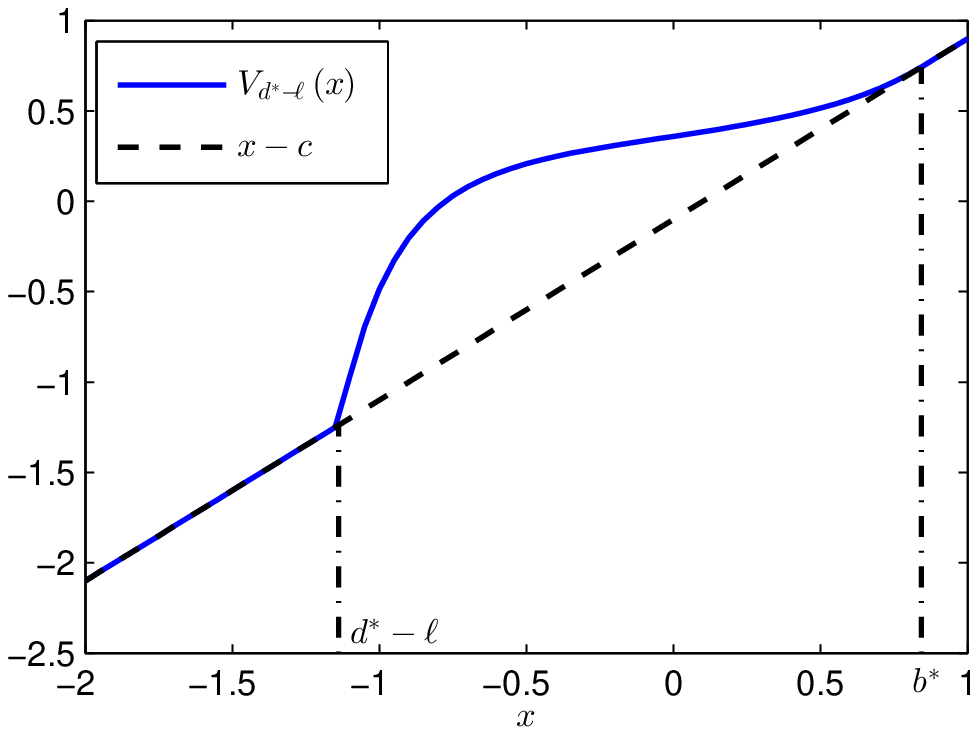}

\caption{\small{(Left) The optimal entry value function $\mathcal{J}_\ell(x)$ dominates the reward function $\mathcal{V}_\ell(x)  - x - \hat{c}$, and they coincide for  $x\le d^*$. (Right) For the exit problem,   the stop-loss level is $d^*-\ell$ and  the optimal liquidation level is $b^*$.  }}
\label{fig:rela_SL}
\end{center}\end{figure}

\subsection{Concluding Remarks}
Other extensions include adapting our double optimal stopping problem to the exponential OU,  Cox-Ingorsoll-Ross (CIR), or other underlying dynamics, and to countable number of trades \citep{zervos2011buy, zhang2008trading}.   Alternatively, one can model asset prices by specifying the dynamics of the dividend stream.  For instance, \cite{scheinkman2003overconfidence} study the optimal timing to trade between two speculative traders with different beliefs on the mean-reverting (OU) dividend dynamics.   Other than  trading of risky assets, it is also useful  to  study the timing to buy/sell derivatives written on  a mean-reverting underlying   (see e.g.  \cite{leungLiu} and \cite{LeungShirai}). For all these applications, it is natural to examine the optimal stopping problems over a finite horizon although explicit solutions are less available.

\appendix
\section{Appendix}
%

\noindent \textbf{A.1 ~Proof of Theorem \ref{thm:V} (Optimality of $V$).}\,

Since $\tau_a\wedge\tau_b\in \setT$, we have $V(x) \geq \sup_{\{a,b: a \leq x \leq b\}} \E_x\{e^{-r(\tau_a\wedge\tau_b)}h(X_{\tau_a\wedge\tau_b})\}=G(x)W(\psi(x))$.

To show the reverse inequality,  we first show that $G(x)W(\psi(x)) \geq \E_x\{e^{-r(t\wedge\tau)}G(X_{t\wedge\tau})W(\psi(X_{t\wedge\tau}))\}$, for $\tau\in\setT$ and $t\geq 0$. The concavity of  $W$ implies that,   for any   fixed $y$, there exists an affine function $L_y(z):=m_y z +c_y$ such that $L_y(z) \geq W(z)$ and $L_y(y)=W(y)$ at $z=y$, where $m_y$ and $c_y$ are both constants depending on  $y$.  This leads to the  inequality
\begin{align}
\E_x\{e^{-r(t\wedge\tau)}G(X_{t\wedge\tau})W(\psi(X_{t\wedge\tau}))\} &\leq \E_x\{e^{-r(t\wedge\tau)}G(X_{t\wedge\tau})L_{\psi(x)}(\psi(X_{t\wedge\tau}))\} \notag\\
&= m_{\psi(x)}\E_x\{e^{-r(t\wedge\tau)}G(X_{t\wedge\tau})\psi(X_{t\wedge\tau})\} + c_{\psi(x)} \E_x\{e^{-r(t\wedge\tau)}G(X_{t\wedge\tau})\}\notag\\
&= m_{\psi(x)}\E_x\{e^{-r(t\wedge\tau)}F(X_{t\wedge\tau})\}+ c_{\psi(x)} \E_x\{e^{-r(t\wedge\tau)}G(X_{t\wedge\tau})\}\notag\\
&= m_{\psi(x)}F(x) + c_{\psi(x)} G(x)\label{1}\\
&= G(x) L_{\psi(x)}(\psi(x))\notag\\
&= G(x)W(\psi(x)),\label{GWsuper}
\end{align}
where   $\eqref{1}$ follows from the martingale property of $(e^{-rt}F(X_t))_{t\ge0}$ and $(e^{-rt}G(X_t))_{t\ge0}$.

By  \eqref{GWsuper} and the fact that $W$ majorizes $H$, it follows that
\begin{align}
G(x) W(\psi(x)) & \geq \E_x\{e^{-r(t\wedge\tau)}G(X_{t\wedge\tau})W(\psi(X_{t\wedge\tau}))\}\notag\\
& \geq \E_x\{e^{-r(t\wedge\tau)}G(X_{t\wedge\tau})H(\psi(X_{t\wedge\tau}))\} = \E_x\{e^{-r(t\wedge\tau)}h(X_{t\wedge\tau})\}.\label{abcd}
\end{align}
Maximizing \eqref{abcd} over all $\tau \in \setT$ and $t\ge 0$ yields that  $G(x)W(\psi(x)) \geq V(x)$.
$\scriptstyle{\blacksquare}$

\noindent \textbf{A.2 ~Proof of Lemma \ref{lm:HOU} (Properties of $H$).}\, The continuity and twice differentiability of $H$ on $(0,+\infty)$ follow directly from those of $h$, $G$ and $\psi$. To show the continuity of $H$ at $0$, since $H(0)=\lim_{x\to-\infty}\frac{(x-c)^+}{G(x)}=0$, we only need to show that $\lim_{y\rightarrow 0} H(y) =0$. Note that $y = \psi(x) \rightarrow 0$, as $x\to -\infty$. Therefore,
\begin{align*}
\lim_{y\rightarrow 0}\limits H(y) = \lim_{x \to -\infty}\limits \frac{h(x)}{G(x)} = \lim_{x \to -\infty}\limits \frac{x-c}{G(x)} =  \lim_{x \to -\infty}\limits\frac{1}{G'(x)} = 0.
\end{align*}
We conclude that $H$ is also continuous at $0$.

\noindent(i) One can show that $\psi(x) \in (0,+\infty)$ for $x\in \R$ and is a strictly increasing function. Then property (i) follows directly from the fact that $G(x) > 0$.

\noindent(ii) By the definition of $H$,
\begin{align*}
H'(y) = \frac{1}{\psi'(x)} (\frac{h}{G})'(x) = \frac{h'(x)G(x) - h(x)G'(x)}{\psi'(x)G^2(x)}, \quad y=\psi(x).
\end{align*}

Since both $\psi'(x)$ and $G^2(x)$ are positive, we only need to determine the sign of $h'(x)G(x) - h(x)G'(x) = G(x) - (x-c)G'(x)$.

Define $u(x) := (x-c) - \frac{G(x)}{G'(x)}$. $u(x)+c$ is the intersecting point at $x$ axis of the tangent line of $G(x)$. Since $G(\cdot)$ is a positive, strictly decreasing and convex function, $u(x)$ is strictly increasing and $u(x)<0$ as $x\to -\infty$. Also, note that
\begin{align*}
&u(c) = - \frac{G(c)}{G'(c)} >0,\\
&u(L^*) = (L^*-c) - \frac{G(x)}{G'(x)} = \frac{\mu}{r}(\theta - L^*) - \frac{G(L^*)}{G'(L^*)} = -\frac{\sigma^2}{2r}\frac{G''(L^*)}{G'(L^*)}>0.
\end{align*}
Therefore, there exists a unique root  $x^*$ that solves $u(x) = 0$, and $x^*< c\wedge L^*$, such that
\begin{align*}
G(x) - (x-c)G'(x)\begin{cases}
< 0 &\, \textrm{ if }\, x \in (-\infty, x^*),\\
>0 &\, \textrm{ if }\, x \in (x^*,+\infty).
\end{cases}
\end{align*}
Thus $H(y)$ is strictly decreasing if $y \in (0, \psi(x^*))$, and increasing otherwise.

\noindent(iii) By the definition of $H$,
\begin{align*}
H''(y) = \frac{2}{\sigma^2G(x)(\psi'(x))^2}(\L-r)h(x),\quad y=\psi(x).
\end{align*}
Since $\sigma^2, G(x)$ and $(\psi'(x))^2$ are all positive, we only need to determine the sign of $(\L-r)h(x)$:
\begin{align*}
(\L - r)h(x) = \mu(\theta - x) - r(x-c) = (\mu\theta+rc)-(\mu+r)x
&\begin{cases}
\geq 0 &\, \textrm{ if }\, x \in (-\infty, L^*],\\
\leq 0 &\, \textrm{ if }\, x \in [L^*, +\infty).
\end{cases}
\end{align*}
Therefore, $H(y)$ is convex if $y \in (0, \psi(L^*)]$, and concave otherwise.
$\scriptstyle{\blacksquare}$

%
%

\noindent \textbf{A.3 ~Proof of Lemma \ref{lm:hatHOU} (Properties of $\hat{H}$).} We first show that $V(x)$ and $\hat{h}(x)$ are twice differentiable everywhere, except for $x=b^*$.  Recall that
\begin{align*}V(x) =
\begin{cases}
(b^*-c)\frac{F(x)}{F(b^*)} &\, \textrm{ if }\, x\in (-\infty,b^*),\\
x-c &\, \textrm{ otherwise},
\end{cases}
\quad \textrm{and} \quad \hat{h}(x) = V(x)-x-\hat{c}.
\end{align*}
Therefore, it follows from   \eqref{eq:solvebOU} that
\begin{align*}V'(x) =
\begin{cases}
(b^*-c)\frac{F'(x)}{F(b^*)} = \frac{F'(x)}{F'(b^*)} &\, \textrm{ if }\, x\in (-\infty,b^*),\\
1 &\, \textrm{ if }\, x\in (b^*,+\infty),
\end{cases}
\end{align*}
which implies that $V'(b^*-) = 1 = V'(b^*+).$ Therefore, $V(x)$ is differentiable everywhere and so is $\hat{h}$. However, $V(x)$ is not twice differentiable since
\begin{align*}V''(x) = 
\begin{cases}
\frac{F''(x)}{F'(b^*)} &\, \textrm{ if }\, x\in (-\infty,b^*), \\
0 &\, \textrm{ if }\, x\in (b^*,+\infty),
\end{cases}
\end{align*}
and $V''(b^*-) \neq V''(b^*+)$.
Consequently, $\hat{h}(x) = V(x)-x-\hat{c}$ is  not twice differentiable at $b^*$.

The twice differentiability of $\hat{G}$ and $\hat{\psi}$ are straightforward. The continuity and differentiability of $\hat{H}$ on $(0,+\infty)$ and twice differentiability on $(0,\hat{\psi}(b^*)) \cup (\hat{\psi}(b^*),+\infty)$ follow directly. Observing that $\hat{h}(x)>0$ as $x\to -\infty$, $\hat{H}$ is also continuous at $0$ by definition. We now establish the properties of $\hat{H}$.

\noindent(i) First we prove the value of $\hat{H}$ at $0$:
\begin{align*}
\hat{H}(0) = \lim_{x \to -\infty}\limits \frac{(\hat{h}(x))^+}{\hat{G}(x)} =  \limsup_{x\to-\infty}\limits\frac{\frac{(b^*-c)}{F(b^*)}F(x)-x-\hat{c}}{\hat{G}(x)} = \limsup_{x\to-\infty}\limits\frac{\frac{(b^*-c)}{F(b^*)}F'(x)-1}{\hat{G}'(x)}=0.
\end{align*}

Next, observe that $\lim_{x\to -\infty}\hat{h}(x) = +\infty$ and $\hat{h}(x) = -(c+\hat{c})$, for $x \in [b^*,+\infty)$. Since $F'(x)$ is  strictly increasing   and  $F'(x) >0$ for $x\in\R$, we have, for $x < b^*$,
\[\hat{h}'(x) = V'(x) -1 = \frac{F'(x)}{F'(b^*)} -1< \frac{F'(b^*)}{F'(b^*)}-1 =0,\]
which implies that $\hat{h}(x)$ is strictly decreasing for $x\in (-\infty,b^*)$. Therefore, there exists a unique solution $\bar{d}$ to $\hat{h}(x) = 0$, and $\bar{d}<b^*$, such that $\hat{h}(x) >0$ if $x \in (-\infty, \bar{d})$ and $\hat{h}(x) <0$ if $x \in (\bar{d},+\infty)$.
It is trivial that $\hat{\psi}(x) \in (0,+\infty)$ for $x\in \R$ and is a strictly increasing function. Therefore, along with the fact that $\hat{G}(x) > 0$, property (i) follows directly.

\noindent(ii) With $y=\hat{\psi}(x)$, for $x>b^*$,
\begin{align*}
\hat{H}'(y) = \frac{1}{\hat{\psi}'(x)} (\frac{\hat{h}}{\hat{G}})'(x) = \frac{1}{\hat{\psi}'(x)}(\frac{-(c+\hat{c})}{\hat{G}(x)})'=\frac{1}{\hat{\psi}'(x)}\frac{(c+\hat{c})\hat{G}'(x)}{\hat{G}^2(x)} < 0,
\end{align*}
since $\hat{\psi}'(x)>0$, $\hat{G}'(x)<0$, and $\hat{G}^2(x)>0$. Therefore, $\hat{H}(y)$ is strictly decreasing for $y > \hat{\psi}(b^*)$.

\noindent(iii) By the definition of $\hat{H}$,
\begin{align*}
\hat{H}''(y) = \frac{2}{\sigma^2\hat{G}(x)(\hat{\psi}'(x))^2}(\L-\hat{r})\hat{h}(x),\quad y=\hat{\psi}(x).
\end{align*}
Since $\sigma^2$, $\hat{G}(x)$ and $(\hat{\psi}'(x))^2$ are all positive, we only need to determine the sign of $(\L - \hat{r})\hat{h}(x)$:
\begin{align*}
(\L - \hat{r})\hat{h}(x) &= \half \sigma^2 V''(x) + \mu(\theta -x) V'(x) - \mu (\theta-x) - \hat{r}(V(x) -x -\hat{c})\\
& = \begin{cases}
(r-\hat{r})V(x) + (\mu+\hat{r})x-\mu\theta +\hat{r}\hat{c} &\, \textrm{ if }\, x < b^*,\\
\hat{r}(c+\hat{c})>0 &\, \textrm{ if }\, x > b^*.
\end{cases}
\end{align*}
To determine the sign of $(\L - \hat{r})\hat{h}(x)$ in $(-\infty,b^*)$, first note that $[(\L - \hat{r})\hat{h}](x)$ is a strictly increasing function in $(-\infty,b^*)$, since $V(x)$ is a strictly increasing function and $r\geq \hat{r}$ by assumption. Next note that for $x\in [L^*,b^*)$,
\begin{align*}
(\L - \hat{r})\hat{h}(x) &=(r-\hat{r})V(x) + (\mu+\hat{r})x-\mu\theta +\hat{r}\hat{c}\\
&\geq (r-\hat{r})(x-c) + (\mu+\hat{r})x-\mu\theta +\hat{r}\hat{c}\\
& = (r+\mu)x - (\mu\theta+rc)+\hat{r}(c+\hat{c})\\
& \geq (r+\mu)L^* - (\mu\theta+rc)+\hat{r}(c+\hat{c})= \hat{r}(c+\hat{c}) >0.
\end{align*}
Also, note that $(\L - \hat{r})\hat{h}(x)\! \to\! -\infty$ as $x\!\to\! -\infty$. Therefore, $(\L - \hat{r})\hat{h}(x) <0$ if $x \in (-\infty, \underline{b})$ and  $(\L - \hat{r})\hat{h}(x) >0$ if $x \in (\underline{b}, +\infty)$ with $\underline{b} < L^*$ being the break-even point.
From this, we conclude property (iii).
%
 $\scriptstyle{\blacksquare}$


\noindent \textbf{A.4 ~Proof of Lemma \ref{lm:hatHOUSL} (Properties of $\hat{H}_L$).}\,
(i) The continuity of  $\hat{H}_L(y)$ on $(0,+\infty)$  is implied by  the continuities of $\hat{h}_L$,  $\hat{G}$ and $\hat{\psi}$.  The continuity of $\hat{H}_L(y)$ at $0$ follows from
\begin{align*}\hat{H}_L(0)&=\lim_{x\to -\infty}\frac{(\hat{h}_L(x))^+}{\hat{G}(x)} = \lim_{x\to -\infty}\frac{0}{\hat{G}(x)}=0,\\
\lim_{y\rightarrow 0}\hat{H}_L(y) &= \lim_{x\to -\infty}\frac{\hat{h}_L}{\hat{G}}(x) = \lim_{x\to -\infty}\frac{-(c+\hat{c})}{\hat{G}(x)} = 0,
\end{align*}where   we have used that  $y = \hat{\psi}(x) \rightarrow 0$ as $x \to -\infty$.

Furthermore,   for $x \in (-\infty, L]\cup [b_L^*,+\infty)$, we have $V_L(x) = x-c$, and thus, $\hat{h}_L(x)=-(c+\hat{c})$. Also, with the facts that $\hat{\psi}(x)$ is a strictly increasing function and $\hat{G}(x) > 0$, property (i) follows.

\noindent (ii)  By the definition of $\hat{H}_L$, since $\hat{G}$ and $\hat{\psi}$ are differentiable everywhere, we only need to show the differentiability of $V_L(x)$. To this end,
 $V_L(x)$ is differentiable at $b_L^*$ by \eqref{VOUSLsol}-\eqref{eq:solvebOUSL}, but not at $L$. Therefore, $\hat{H}_L$ is differentiable for $y \in (0,\hat{\psi}(L)) \cup (\hat{\psi}(L), +\infty)$.

In view of the facts that $\hat{G}'(x) <0$, $\hat{\psi}'(x)>0$, and $ \hat{G}^2(x)>0$, we have for $x\in (-\infty, L)\cup [b_L^*,+\infty)$,
\begin{align*}
\hat{H}_L'(y) = \frac{1}{\hat{\psi}'(x)} (\frac{\hat{h}_L}{\hat{G}})'(x) =\frac{1}{\hat{\psi}'(x)} (\frac{-(c+\hat{c})}{\hat{G}(x)})' = \frac{(c+\hat{c})\hat{G}'(x)}{\hat{\psi}'(x)\hat{G}^2(x)} <0.
\end{align*}
Therefore, $\hat{H}_L(y)$ is strictly decreasing for $y \in (0, \hat{\psi}(L))\cup [\hat{\psi}(b_L^*),+\infty)$.

\noindent (iii) Both $\hat{G}$ and $\hat{\psi}$ are twice differentiable everywhere, while $V_L(x)$ is twice differentiable everywhere except at $x=L$ and $b^*$, and so is $\hat{h}_L(x)$. Therefore,
$\hat{H}_L(y)$ is twice differentiable on $(0,\hat{\psi}(L)) \cup (\hat{\psi}(L),\hat{\psi}(b^*)) \cup (\hat{\psi}(b^*),+\infty)$.

To determine the convexity/concavity of $\hat{H}_L$, we look at the second order derivative:
\begin{align*}
\hat{H}_L''(y) = \frac{2}{\sigma^2\hat{G}(x)(\hat{\psi}'(x))^2}(\L-\hat{r})\hat{h}_L(x),
\end{align*}whose sign is determined by
\begin{align*}
(\L - \hat{r})\hat{h}_L(x) &= \half \sigma^2 V_L''(x) + \mu(\theta -x) V_L'(x) - \mu (\theta-x) - \hat{r}(V_L(x) -x -\hat{c})\\
& = \begin{cases}
(r-\hat{r})V_L(x) + (\mu+\hat{r})x-\mu\theta +\hat{r}\hat{c} &\, \textrm{ if }\, x \in (L, b_L^*),\\
\hat{r}(c+\hat{c}) >0 &\, \textrm{ if }\, x \in (-\infty, L) \cup (b_L^*,+\infty).
\end{cases}
\end{align*}
This implies that $\hat{H}_L$ is convex for $y \in (0, \hat{\psi}(L))\cup (\hat{\psi}(b_L^*),+\infty)$.

On the other hand, the condition $\sup_{x\in\R}\hat{h}_L(x)>0$ implies that $\sup_{y\in [0,+\infty)}\hat{H}_L(y)>0$. By property \eqref{hatHOUSL0} and twice differentiability of $\hat{H}_L(y)$ for $y \in (\hat{\psi}(L),\hat{\psi}(b_L^*))$, there must exist an interval $(\hat{\psi}(\underline{a}_L), \hat{\psi}(\bar{d}_L)) \subseteq (\hat{\psi}(L),\hat{\psi}(b_L^*))$ such  that $\hat{H}_L(y)$ is concave, maximized at $\hat{z}_1 \in (\hat{\psi}(\underline{a}_L), \hat{\psi}(\bar{d}_L))$.

 Furthermore, if  $V_L(x)$ is strictly increasing on   $(L, b_L^*)$,    then   $(\L - \hat{r})\hat{h}_L(x)$ is also strictly increasing. To prove this, we first recall from  Lemma \ref{lm:HOU} that $H(y)$ is strictly increasing and concave on $(\psi(L^*), +\infty)$. By Proposition \ref{prop:bL}, we have $b_L^* < b^*$, which  implies $z_L < z$, and thus, $H'(z_L)>H'(z)$.

 Then, it follows from  \eqref{eq:HzzLiquOU}, \eqref{Wy} and \eqref{eq:WLy} that $W_L'(y)= H'(z_L)> H'(z) = W'(y)$ for $y\in (\psi(L), z_L)$. Next, since $W_L(y) = \frac{V_L}{G}\circ \psi^{-1}(y)$, we have
\begin{align*}
W_L'(y) &= \frac{1}{\psi'(x)} (\frac{V_L}{G})'(x) = \frac{1}{\psi'(x)} (\frac{V_L'(x)G(x) - V_L(x)G'(x)}{G^2(x)}).
\end{align*}
The same holds for $W'(y)$ with $V(x)$ replacing $V_L(x)$. As both $\psi'(x)$ and $G^2(x)$ are positive, $W_L'(y)>  W'(y)$ is equivalent to $V_L'(x)G(x) - V_L(x)G'(x) > V'(x)G(x) - V(x)G'(x)$. This implies that
\begin{align*}
V_L'(x) - V'(x) = -\frac{G'(x)}{G(x)}(V(x) - V_L(x)) >0,
\end{align*}
since $G(x)>0$, $G'(x)<0$, and $V(x)>V_L(x)$.  Recalling that $V'(x)>0$,   we have established that $V_L(x)$ is a strictly increasing function, and so is $(\L - \hat{r})\hat{h}_L(x)$. As we have shown the existence of an interval $(\hat{\psi}(\underline{a}_L), \hat{\psi}(\bar{d}_L)) \subseteq (\hat{\psi}(L),\hat{\psi}(b_L^*))$ over which $\hat{H}(y)$ is concave, or equivalently $(\L - \hat{r})\hat{h}_L(x)<0$ with $x = \hat{\psi}^{-1}(y)$. Then by the strictly increasing property of $(\L - \hat{r})\hat{h}_L(x)$, we conclude $\underline{a}_L = L$ and $\bar{d}_L \in (L, b_L^*)$ is the unique solution to $(\L - \hat{r})\hat{h}_L(x)=0$, and
\begin{align*}
(\L - \hat{r})\hat{h}_L(x) \begin{cases}
<0 &\, \textrm{ if }\, x \in (L, \bar{d}_L),\\
>0 &\, \textrm{ if }\, x \in (-\infty, L)\cup (\bar{d}_L, b_L^*) \cup (b_L^*,+\infty).
\end{cases}
\end{align*}
Hence, we conclude the convexity and concavity of the function $\hat{H}_L$.
$\scriptstyle{\blacksquare}$

\bibliographystyle{apa}
\linespread{-0.3}
\begin{small}
\bibliography{mybib_2012}
\end{small}
\end{document}